\documentclass[final, runningheads,envcountsect,envcountsame]{llncs}

\pdfoutput=1
\usepackage{ifthen}
\newboolean{spellingmode} \setboolean{spellingmode}{false}
\newboolean{cropmargins} \setboolean{cropmargins}{false}
\newboolean{arxivmode} \setboolean{arxivmode}{true}
\usepackage{ifdraft}
\ifdraft{\setboolean{cropmargins}{true}}{}
\ifthenelse{\boolean{arxivmode}}{\setboolean{cropmargins}{true}}{}

\ifspellingmode
\fi

\usepackage{amsmath}
\usepackage{amssymb}
\usepackage{mathtools}
\usepackage{xspace}
\newboolean{showappendix}
\usepackage{environ}
\usepackage{etoolbox}
\usepackage{rotating}
\usepackage{booktabs}
\usepackage{mathabx}
\usepackage{array}
\ifthenelse{\boolean{arxivmode}}{\usepackage{orcidlink}
\def\orcidID#1{\orcidlink{#1}}}

\setboolean{showappendix}{false}

\usepackage[T1]{fontenc}
\usepackage{graphicx}
\usepackage{tikz}
\usetikzlibrary{cd,calc,fit,positioning,decorations.markings,arrows.meta}
\ifthenelse{\boolean{cropmargins}}{
\usepackage[marginparwidth=2cm,nohead,nofoot,
            headsep = 1cm,
            footskip = 1cm,
            text={12.2cm,19.3cm},
            paperwidth=16.2cm,
            paperheight=23.3cm,
            marginparsep=1mm,
            marginparwidth=1.9cm,
            footskip=1cm,
            centering
            ]{geometry}
}{}

\newsavebox{\kleisliarrow}
\savebox{\kleisliarrow}{%
\begin{tikzpicture}[
      baseline=(arrow.base),
      inner sep=0mm,
      outer sep=0mm,
      ]
      \node[draw=none,
      anchor=base,
      overlay,
      inner sep=0,
      outer sep=0,
      minimum height=1em,
      ] (arrow) {$\phantom{\longrightarrow}$};
      \begin{scope}[even odd rule,overlay]
        \clip  ($ (arrow.south) !.68! (arrow.north)$) circle (0.17em)
           (arrow.north west) rectangle (arrow.south east);
      \node[draw=none,
      anchor=base,
      inner sep=0,
      outer sep=0,
      ] (arrow) {$\longrightarrow$};
    \end{scope}
    \draw[fill=none,overlay] ($ (arrow.south) !.68! (arrow.north)$) circle (0.15em);
    \draw[use as bounding box,draw=none] (arrow.north west) rectangle (arrow.south east);
  \end{tikzpicture}}

\renewcommand{\arg}{\ensuremath{\_\!\!\!\_}}

\newsavebox{\kleislidot}
\savebox{\kleislidot}{%
\begin{tikzpicture}[baseline=0pt,outer sep=0pt]
    \draw[fill=white] (0,0) circle (2pt);
  \end{tikzpicture}}

\tikzstyle{kleisli}=[
"{\usebox{\kleislidot}}"{red,anchor=center,font=\normalsize,solid,pos=0.5},
outer sep = 1pt, %
]

\newsavebox{\mypullbackcorner}%
\sbox{\mypullbackcorner}{%
\begin{tikzpicture}
    \draw[-] (0,0) -- (.5em,.5em) -- (0,1em);
\end{tikzpicture}%
}
\newcommand{\pullbackangle}[2][]{\arrow[phantom,to path={
                     -- ($ (\tikztostart)!1cm!#2:([xshift=8cm]\tikztostart) $)
                        node[anchor=west,pos=0.0,rotate=#2,
                        inner xsep = 0]
                        {\begin{tikzpicture}[minimum
                        height=1mm,baseline=0,#1]
    \draw[-] (0,0) -- (.5em,.5em) -- (0,1em);
                        \end{tikzpicture}}}]{}}

\tikzstyle{shiftarr}=[
        rounded corners,%
        to path={--([#1]\tikztostart.center)
                     -- ([#1]\tikztotarget.center) \tikztonodes
                     -- (\tikztotarget)},
]

\tikzset{
      vertex/.style={
        fill=black,
        shape=circle,
        outer sep = 1mm,
        inner sep = 0mm,
        minimum size=2mm,
      },
      set/.style={
        draw=black!20,
        fill=none,
        line width=1mm,
        rounded corners=2mm,
        inner sep = 2mm,
      },
      myedge/.style={
        ->,
        draw=black,
        every node/.append style={
          shape=circle,
          inner sep=1pt,
          fill=black,
          text=white,
          draw=none,
          sloped,
          minimum width=5mm,
        },
      },
}

\newcommand{\xtransto}[1]{%
\mathrel{%
  \begin{tikzpicture}[baseline=(a.base), weighted system]
    \node[font=\scriptsize,inner sep=1mm,inner ysep=1pt] (label) {\ensuremath{#1}} ;
    \node[anchor=west] (a) at (label.south west) {\phantom{X}};
    \node[anchor=east] (b) at ([xshift=2pt]label.south east) {};
    \draw[transition] (a.west) to (b.east);
  \end{tikzpicture}
}
}
\newcommand{\xtranstostar}[1]{%
\mathrel{%
  {\xtransto{#1}{\!\!}^{{\color{transitioncolor}\star}}\,}%
}%
}

\newcommand{\transto}{
\mathord{%
  \begin{tikzpicture}[baseline=(a.base), weighted system]
    \node (a) {\phantom{X}};
    \node (b) at (4pt,0) {};
    \coordinate (arrow start) at ([yshift=-1pt]a.west);
    \draw[transition] (arrow start) to (b.east |- arrow start);
  \end{tikzpicture}
}
}

\usetikzlibrary{external}
\ifthenelse{\pdfshellescape=1}{
\tikzexternalize[disable dependency files]
}{}
\tikzset{external/up to date check=diff}

\tikzsetexternalprefix{diagrams/}
\makeatletter
\tikzifexternalizing{%
\def\@enddocumenthook{}%
}{}
\makeatother

\newenvironment{proofappx}[2][Proof of]{%
\subsection*{#1~\autoref{#2}}%
\addcontentsline{toc}{subsection}{#1~\autoref{#2}}%
\begin{proof}%
}{
\end{proof}%
}

\def\temp{&} \catcode`&=\active \let&=\temp

\tikzcdset{
arrow style=tikz,
diagrams={>={Straight Barb[scale=0.8]}}
}

\newcommand{\mytikzcdcontext}[2]{
  \begin{tikzpicture}[baseline=(mainnode.base)]
    \node (mainnode) [inner sep=0, outer sep=0] {\begin{tikzcdold}[#2]
        #1
    \end{tikzcdold}};
  \end{tikzpicture}}

\NewEnviron{tikzcd}[1][]{%
\def\myargs{#1}%
\edef\mydiagram{\noexpand\mytikzcdcontext{\expandonce\BODY}{\expandonce\myargs}}%
\mydiagram%
}

\makeatletter
\newenvironment{notheorembrackets}{%
\csdef{@spopargbegintheorem}##1##2##3##4##5{\trivlist%
      \item[\hskip\labelsep{##4##1\ ##2}]{##4{##3}\@thmcounterend\ }##5}%
    }{%
\csdef{@spopargbegintheorem}##1##2##3##4##5{\trivlist%
      \item[\hskip\labelsep{##4##1\ ##2}]{##4(##3)\@thmcounterend\ }##5}%
    }
\makeatother

\usepackage[footnote,marginclue,nomargin]{fixme}
\FXRegisterAuthor{tw}{atw}{TW} %
\FXRegisterAuthor{jd}{ajd}{JD} %

\usepackage{hyperref}
\usepackage{twshowkeys}

\title{Weighted and Branching Bisimilarities from\texorpdfstring{\\}{ }Generalized Open Maps}
\makeatletter
\hypersetup{
  final,
  hidelinks,
  pdftitle={\@title},
  pdfauthor={J\'er\'emy Dubut, Thorsten Wi{\ss}mann},
}
\makeatother

\newcommand{\etal}[1][.]{et\ al#1}
\newcommand{\M}{\ensuremath{\mathbb{M}}}
\newcommand{\E}{\ensuremath{\mathbb{E}}}
\newcommand{\bbS}{\ensuremath{\mathbb{S}}}
\newcommand{\Path}{\ensuremath{\mathbb{P}}}
\newcommand{\LTS}{\ensuremath{\mathsf{LTS}}}
\newcommand{\Set}{\ensuremath{\mathsf{Set}}}
\newcommand{\Pos}{\ensuremath{\mathsf{Pos}}}
\newcommand{\dual}{\ensuremath{\mathsf{op}}}
\newcommand{\JE}{\ensuremath{J_{\E}}}
\newcommand{\JS}{\ensuremath{J_{\bbS}}}
\newcommand{\hateq}{\ensuremath{\mathrel{\widehat{=}}}}
\newcommand{\WLTS}{\ensuremath{\mathsf{WLTS}}}
\newcommand{\descto}[3][]{\arrow[phantom]{#2}[#1]{\text{\footnotesize{}#3}}}
\newcommand{\commutes}{\ensuremath{\circlearrowright}\xspace}

\newcommand{\pow}{\ensuremath{\mathcal{P}}}
\newcommand{\C}{\ensuremath{\mathcal{C}}}

\newcommand{\FSim}{\ensuremath{\operatorname{\text{\upshape\sf LCoalg}}\xspace}}
\newcommand{\Dist}{\ensuremath{\mathcal{D}\xspace}}
\newcommand{\Subdist}{\ensuremath{\mathcal{D}_{\leq1}\xspace}}
\newcommand{\bbP}{\ensuremath{\mathbb{P}}}
\newcommand{\id}{\ensuremath{\operatorname{\textnormal{id}}\xspace}}
\newcommand{\downinclusion}{\begin{turn}{270}{$\sqsubseteq$}\end{turn}}
\newcommand{\ob}[1]{\ensuremath{\operatorname{\textnormal{ob}}(#1)\xspace}}
\newcommand{\Nat}{\ensuremath{\mathbb{N}}}
\newcommand{\runend}{\ensuremath{\text{end}}}
\newcommand{\Real}{\ensuremath{\mathbb{R}}}
\newcommand{\Realp}{\ensuremath{\mathbb{R}_+}}
\newcommand{\weighted}[2]{\ensuremath{{#1}^{(#2)}}}
\newcommand{\set}[2][]{%
  \ifthenelse{\equal{#2}{}}{%
    \ensuremath{{#1\emptyset}}%
  }{%
    \ensuremath{{#1\{#2#1\}}}%
  }%
}

\makeatletter
\gdef\gsetlength#1#2{%
  \begingroup
    \setlength\skip@{#2}%
    \global#1=\skip@    %
  \endgroup             %
}
\makeatother

\colorlet{transitioncolor}{blue!30!black!80!white}

\tikzset{
  weighted system/.style={
    x=1.2cm,
    y=1.2cm,
    every node/.append style={inner sep=1pt},
    state/.style={
      execute at begin node=$,%
      execute at end node=$%
    },
    every loop/.style={looseness=10, min distance=4mm},
    transition/.style={
      line width=1pt,
      -{Triangle[length=4pt,width=4pt]},
      draw=transitioncolor,
      every node/.append style={
        font=\upshape\footnotesize,
        inner sep=2pt,
        above,
        execute at begin node=$,%
        execute at end node=$%
        },
    },
    frame decoration/.style={
      line width=1pt,
      draw=black!30,
      rounded corners=2pt,
    },
  },
}

\newlength{\mymaxradius}
\newlength{\mycurrentradius}
\newcommand{\weightedsystem}[2][r]{%
  \begin{tikzpicture}[baseline=(r.base),weighted system]
  \node (r) at (0,0) {\ensuremath{#1}};
  \node[inner xsep=1pt,anchor=center] at ([xshift=-1.5mm]r.west) (gapOnTheWest) {};
  \draw[transition] (gapOnTheWest.west) to (r.west);
  #2
  \draw[frame decoration]
    ([xshift=-2pt,yshift=2pt]current bounding box.north west) rectangle
    ([xshift=2pt,yshift=-2pt]current bounding box.south east);
  \end{tikzpicture}%
}
\newcommand{\simplepath}[2][20]{%
  \begin{tikzpicture}[baseline=(r.base),weighted system]
  \node (r) at (0,0) {$r$};
  \setlength{\mymaxradius}{15mm}
  \pgfmathsetmacro{\succcount}{0}
  \foreach \weight/\name in {#2} {
      \pgfmathparse{\succcount+1}
      \global\let\succcount\pgfmathresult
      \settowidth{\mycurrentradius}{\ensuremath{\weight\qquad{}x}}
      \ifthenelse{\lengthtest{\mycurrentradius > \mymaxradius}}{
        \gsetlength{\mymaxradius}{\the\mycurrentradius}
      }{}
  }
  \pgfmathsetmacro{\spanangle}{#1}
  \foreach \weight/\name [count=\i] in {#2} {
    \pgfmathsetmacro{\angle}{(\succcount < 2) ? 0 :
      ((\spanangle / 2) - ((\i-1) * \spanangle / (max(2,\succcount) -1 )))
    }
    \node (node) at (\angle:\the\mymaxradius) {$\name $};
    \ifthenelse{\i<\succcount \OR \i=1}{
      \draw[transition] (r) to node[sloped,above,font=\upshape\footnotesize] {\weight} (node);
    }{
      \draw[transition] (r) to node[sloped,below,font=\upshape\footnotesize] {\weight} (node);
    }
  }
  \node[inner xsep=1pt,anchor=center] at ([xshift=-1.5mm]r.west) (gapOnTheWest) {};
  \draw[transition] (gapOnTheWest.west) to (r.west);
  \draw[frame decoration]
    ([xshift=-1pt,yshift=1pt]current bounding box.north west) rectangle
    ([xshift=1pt,yshift=-1pt]current bounding box.south east);
\end{tikzpicture}}

\spnewtheorem{assumption}[theorem]{Assumption}{\bfseries}{\itshape}

\begin{document}
\titlerunning{Weighted Bisimilarity from Generalized Open Maps}
\author{%
J\'er\'emy Dubut\inst{1}\orcidID{0000-0002-2640-3065}
\and
Thorsten Wi{\ss}mann\inst{2}\orcidID{0000-0001-8993-6486}
}
\authorrunning{J.~Dubut and T.~Wi{\ss}mann}
\institute{
National Institute of Advanced Industrial Science and Technology, Tokyo, Japan
\email{jeremy.dubut@aist.go.jp} \and
Radboud University, Nijmegen, the Netherlands
  \email{thorsten.wissmann@ru.nl}}

\maketitle              %
\begin{abstract}
In the open map approach to bisimilarity, the paths and their runs in a given 
state-based system are the first-class citizens, and bisimilarity becomes a 
derived notion. While open maps were successfully used to model 
bisimilarity in non-deterministic systems, the approach fails to describe quantitative 
system equivalences such as probabilistic bisimilarity.
In the present work, we see that this is indeed impossible and we thus 
generalize the notion of open maps to also accommodate
weighted and probabilistic bisimilarity.
Also, extending the notions of strong path and path bisimulations into 
this new framework, we show that branching bisimilarity can be 
captured by this extended theory and that it can be viewed as 
the history preserving restriction of weak bisimilarity.

\end{abstract}

\section{Introduction}

The theory of open maps is a categorical framework to reason about 
systems and their bisimilarities~\cite{joyal96}.
Given a category of systems and a description of the shape of the 
executions and how to extend them, open maps are morphisms with lifting 
properties with respect to those extensions.
Intuitively, open maps are morphisms which preserve and reflect transitions of 
systems, that is, they are morphisms whose graphs are bisimulations.
The theory %
covers various classical notions of bisimilarity. For example, 
two LTSs are strongly bisimilar if and only if there is a span of open maps between them. 
Varying the category of models 
and the execution shapes allows describing weak bisimilarity, timed bisimilarity, 
probabilistic Larsen and Skou bisimilarity, and history-preserving bisimilarity of event 
structures (see \cite{joyal96,cheng95,hune98} for examples).

Another categorical framework for bisimilarity is coalgebra \cite{rutten00}.
This time, given a category and an endofunctor describing respectively the 
type of state spaces and the type of transitions, a `system' is understood as
a coalgebra for this functor. Coalgebra homomorphisms are then very similar to 
open maps in spirit: they also are morphisms that preserve and reflect transitions.
This intuition has been made formal by transformations between the categorical 
frameworks in both ways; from open maps to coalgebra
\cite{lasota02}, and conversely \cite{WDKH2019}. 
However, the latter suggests that open maps are only 
adapted to modeling non-deterministic systems and would struggle with other 
types of branchings, such as probabilistic.

In coalgebra, there are no particular difficulties in modeling weighted systems, 
and by extension, discrete probabilistic systems \cite{klin09}.
There is also some work for continuous probabilities, although the theory is 
much more complicated \cite{desharnais02,danos06}.
As we will explain more precisely later, there have been some attempts to do so 
with open maps in \cite{cheng95,desharnais02}, but the result is somewhat disappointing.

Conversely, coalgebra is not adapted to bisimilarities for systems where 
transitions are not history-preserving, that is, for which the behavioral equivalence does 
not just depend on the transitions at a given state, but on 
the whole history of the execution that led to this 
state. That is the case for example for branching bisimilarity \cite{glabbeek96}. 
Branching bisimilarity arose precisely to make weak bisimilarity history-preserving.
In \cite{cheng95}, weak bisimilarity has been described using open maps by carefully 
choosing the underlying category, with a general theory developed in \cite{fiore99} using 
presheaf models. Branching bisimilarity has also been studied using open maps in
\cite{beohar15,beohar19}, but indirectly, through a translation into presheaves.

To resume, the goal of this paper is to capture weighted and branching 
bisimilarities using a generalization of open maps. 
Concretely, the contributions are: 
\begin{enumerate}
	\item a proof that it is impossible to appropriately model probabilistic system using standard
		open maps (Section~\ref{sec:impossibility}),
	\item a faithful extension of the theory of open maps and (strong) path 
		bisimulations (Section~\ref{sec:generalOpenMaps}), 
	\item a generalized open map situation capturing weighted and probabilistic bisimilarities
		(Section~\ref{sec:open-weighted}),
	\item a generalized open map situation where strong path bisimulations correspond to 
		stuttering branching bisimulations, open map bisimilarity to branching 
		bisimilarity, and path bisimulations to weak bisimulations 
		(Section~\ref{sec:weak-bisimulations}).
\end{enumerate}
In addition, Section~\ref{sec:pathCat} introduces some theoretical requirements. 

\section{From path categories to bisimilarity}
\label{sec:pathCat}
Before discussing weighted bisimilarity, let us first recall the main ideas of
modeling bisimilarity via open maps, as introduced by Joyal \etal\ \cite{joyal96}. The definition is parametric in a functor $J\colon \Path\to \M$, from a category $\Path$ of paths
to a category $\M$ of models or systems of interest.
In the prime example, $\M$ is the category of labelled transition systems $\LTS$ as defined next:
\begin{definition}
\label{def:lts}
    For a fixed set $A$ of labels, the category $\LTS$ contains:
    \begin{enumerate}
        \item Objects: a labelled transition system $(X,\transto,x_0)$ is a
            set $X$ of states, a transition relation $\mathord{\transto}\subseteq X\times A\times X$ and
            a distinguished initial state $x_0\in X$. We write $x\xtransto{a} x'$ to denote that 
            $(x,a,x')\in \mathord{\transto}$ and simply refer to the LTS as $X$ if
            $\transto$ and $x_0$ are clear from the context. For disambiguation, we use $\to$ for morphisms and $\transto$ for transitions.
        \item Morphisms: a \emph{functional simulation} $f\colon (X,\transto,x_0)\to (Y,\transto,y_0)$ 
            is a function $f\colon X\to Y$ with $f(x_0) = y_0$ and for all
            $x\xtransto{a} x'$ in $X$, we have $f(x) \xtransto{a} f(x')$.
    \end{enumerate}
\end{definition}

A functional simulation $f\colon X\to Y$ intuitively means that the system $Y$ has
at least the transitions of $X$, but possibly more. A special case of a
functional simulation is the \emph{run} of a word in a system:
\begin{definition}
\label{def:lts-paths}
    For the label set $A$, let $(A^*,\le)$ be the partially ordered set of
    words, ordered by the prefix ordering. The functor $J\colon (A^*,\le)\to \LTS$ sends
    a word $w\in A^*$ to the LTS $Jw = (\{v\mid v\le w\}, \transto, \varepsilon)$
    of all prefixes of $w$
    with $v \xtransto{a} va$ for all $a\in A$, $va\le w$.
\end{definition}
This functor $J$ (or more precisely, its image) is often called \emph{path category}
of $\LTS$:
the possible runs of a word $w\in A^*$ in $(X,\transto,x_0)$ correspond precisely to
the functional simulations $Jw\to (X,\transto,x_0)$ in $\LTS$.

On the abstract level, for a general functor $J\colon \Path\to \M$, we
understand the set of morphisms $r\colon Jw\to X$ for $w\in \Path$ and $X\in \M$
as the runs of the path $w$ in the model $X$. We can already make the trivial
observation that all morphisms $f\colon X\to Y$ in $\M$ preserve runs: given a
run $r\colon Jw\to X$ of some path $w\in \Path$ in $X$, there is a run $f\cdot
r\colon Jw\to Y$ of $w$ in $Y$.

The converse does not hold for a general $f\colon X\to Y$ in $\M$: given a run of $w$ in $Y$, there is not necessarily a run of $w$ in $X$. If $f$ reflects runs, it is called \emph{open}:\!\!\twnote{}
\begin{definition}
    For a functor $J\colon \Path\to \M$, a morphism $f\colon X\to Y$ in $\M$ is called \emph{open} if $f$ satisfies the following \emph{lifting property} for all $e\colon v\to w$ in $\Path$:
    \[
        \text{for all }
        \quad
        \begin{tikzcd}
            Jv
            \arrow{r}[description,inner sep=2pt]{r}
            \arrow{d}[description,inner sep=2pt]{Je}
            \descto{dr}{\commutes}
            & X
            \arrow{d}[description,inner sep=2pt]{f}
            \\
            Jw
            \arrow{r}[description,inner sep=2pt]{s}
            & Y
        \end{tikzcd}
        \quad
        \text{ there is }
        d\colon Jw\to X\text{ with }
        \quad
        \begin{tikzcd}
            Jv
            \arrow{r}[description,inner sep=2pt]{r}
            \arrow{d}[description,inner sep=2pt]{Je}
            \descto[pos=0.15]{dr}{\commutes}
            \descto[pos=0.85]{dr}{\commutes}
            & X
            \arrow{d}[description,inner sep=2pt]{f}
            \\
            Jw
            \arrow{r}[description,inner sep=2pt]{s}
            \arrow[dashed]{ur}[description,inner sep=2pt]{d}
            & Y
        \end{tikzcd}
    \]
\end{definition}
That is, for all commutative squares ($s\cdot Je = f\cdot r$),
there is $d\colon Jw\to X$ in $\M$ that makes both triangles on the right
commute ($f\cdot d = s$ and $d\cdot Je= r$).

By construction, we can only make statements about states that are reachable
via some run. Thus, one often restricts $\M$ beforehand to contain only
models in which all states are reachable from the initial state.

For LTSs in which all states are reachable from the initial state, 
open maps are related to strong bisimulations~\cite{park81}: open maps are precisely 
functions whose graph relation $\{(x,fx) \mid x \in X\}$ is a strong bisimulation.
Reformulated in the context of allegories~\cite{freyd90}, open maps are precisely 
the maps in the allegory of relations that are strong bisimulations. It is then 
natural to recover bisimulations as tabulations of open maps, that is:

\begin{definition}
   For a functor $J\colon \Path\to \M$, we say that two models $X$ and $Y$ 
   are \emph{$J$-bisimilar}, if there exist another model $Z$ and two $J$-open 
   maps $f\colon Z\to X$ and $g\colon Z\to Y$, that is, if there is a span of 
   $J$-open maps between them.
\end{definition}

Of course, $J$-bisimilarity is a reflexive (identities are open maps) and 
symmetric (by permuting $f$ and $g$ in the definition) relation on models, 
but it is not transitive in general. It is when the category $\M$ has 
pullbacks \cite{joyal96}.

Given a functor $J\colon \Path\to \M$, there are more classical ways of defining 
bisimilarities given in \cite{joyal96}. The first one is \emph{(strong) path bisimulations}, 
which are relations on runs (similar to history-preserving bisimulations) 
satisfying the usual bisimilarity conditions. The second one is by using a 
modal logic similar to the Hennessy-Milner theorem. In the case of LTSs with 
strong bisimilarity, all those notions describe the same notion of bisimilarity, but 
that is not true for general $J\colon \Path\to \M$:
\twnote{} 
\jdnote{}
it can only be proved that $J$-bisimilarity implies the 
existence of a (strong) path bisimulation, which itself implies that the two models 
satisfy the same formulas of the modal logic. In \cite{dubut16}, some mild sufficient 
conditions in terms of trees (i.e., colimits of paths in $\M$) are given for those 
three notions to coincide. In particular, all the examples of bisimilarities covered 
by open maps cited earlier satisfy these conditions.

We use coalgebra for uniform statements about state-based systems of
different branching type (including non-deterministic and probabilistic
branching):

\begin{definition}
For an object $1$ of a category $\C$ and an endofunctor 
$F\colon\, \C \to \C$,
a \emph{pointed coalgebra} is a pair of 
morphisms of $\C$ of the form
$
	1 \xrightarrow{~i~} X \xrightarrow{~\xi~} FX.
$
\end{definition}

For example, LTSs can be modeled as pointed coalgebras with $\C = \Set$, $1$ any singleton, 
and $F = \pow(A\times \_)$, where $\pow$ is the power set functor.
The usual notion of morphisms of coalgebras can be spelt out as follows:

~\\
\noindent\medskip
\begin{minipage}{.6\linewidth}
\begin{definition}
A \emph{(proper) homomorphism} of pointed coalgebras from $(X,\xi,i)$ to $(Y,\zeta,j)$ is a
morphism $f\colon\, X \to Y$ of $\C$ such that the diagram on the right commutes.
\end{definition} 
\end{minipage}
\begin{minipage}{.35\linewidth}%
 \hspace*{0pt}\hfill
    \begin{tikzcd}
      1
      \arrow{r}[description,inner sep=2pt]{i}
      \arrow[bend right=40]{rd}[description,inner sep=2pt]{j}
      \descto{dr}{$\commutes$}
      & X
      \arrow{r}[description,inner sep=2pt]{\xi}
      \arrow{d}[description,inner sep=2pt]{f}
      \descto{dr}{$\commutes$}
      & FX
      \arrow{d}[description,inner sep=2pt]{Ff}
      \\
      & Y
      \arrow{r}[description,inner sep=2pt]{\zeta}
      & FY
    \end{tikzcd}
 \end{minipage}
 ~\\

Pointed coalgebras and proper homomorphisms always form a category, 
but in the case of LTSs as described above, this category is not equivalent to the category
$\LTS$. Indeed, proper homomorphisms are not just morphisms that preserve transitions, but
similarly to open maps, they also reflect them. In \cite{WDKH2019}, the authors 
proved that for a large class of endofunctors, whose coalgebras basically are non-deterministic,
proper homomorphisms precisely correspond to $J$-open maps for a certain 
functor $J$.
To model morphisms that are only required to preserve transitions, homomorphisms have to be 
made lax as follows (see \cite{WDKH2019}):\\~\\
\noindent\medskip
\begin{minipage}{.6\linewidth}
\begin{definition}
Assume a
relation $\sqsubseteq$ on every Hom-set $\C(X,FY)$.
A \emph{lax homomorphism} of pointed coalgebras from $(X,\xi,i)$ to $(Y,\zeta,j)$ is a
morphism $f\colon\, X \to Y$ of $\C$ such that the diagram on the right laxly commutes, 
that is, $f\cdot i = j$ and $Ff\cdot\xi \sqsubseteq \zeta\cdot f$ in $\C(X,FY)$.
\end{definition}
\end{minipage}
\begin{minipage}{.35\linewidth}
 \hspace*{0pt}\hfill
    \begin{tikzcd}
      1
      \arrow{r}[description,inner sep=2pt]{i}
      \arrow[bend right=40]{rd}[description,inner sep=2pt]{j}
      \descto{dr}{$\commutes$}
      & X
      \arrow{r}[description,inner sep=2pt]{\xi}
      \arrow{d}[description,inner sep=2pt]{f}
      \descto{dr}{\downinclusion}
      & FX
      \arrow{d}[description,inner sep=2pt]{Ff}
      \\
      & Y
      \arrow{r}[description,inner sep=2pt]{\zeta}
      & FY
    \end{tikzcd}
\end{minipage}
~\\

In the case of the functor $\pow(A\times\_)$, we can consider the pointwise inclusion on 
every Hom-set $\Set(X,\pow(A\times Y))$. With this, pointed coalgebras and lax 
homomorphisms form a category which is isomorphic to the category $\LTS$.
However, it is not true in general that they form a category, as a compatibility of $\sqsubseteq$ 
with the composition is needed as follows:
\begin{definition}
\label{def:partialOrder}
A \emph{partial order on $F$} is a collection of partial orders $\sqsubseteq$, one for each 
Hom-set of the form $\C(X,FY)$ such that
\[
	\forall X\xrightarrow{f_1,f_2} FY,~
  X'\xrightarrow{g} X,~
  Y\xrightarrow{h} Y'\colon ~~
  f_1 \sqsubseteq f_2 ~~\Rightarrow ~~
		Fh\cdot f_1 \cdot g \sqsubseteq Fh\cdot f_2 \cdot g.
\]
This is equivalent to the requirement that the Hom-functor $\C(\arg, F\arg)$ factors through partially ordered sets:
\(
\C(\arg, F\arg)\colon \C^\dual\times \C\to \Pos.
\)
\end{definition}
\begin{remark} \label{remOrderHJ}
The present definition subsumes the definition of order on a Set-functor
established by Hughes and Jacobs~\cite[Def 2.1]{HughesJ04} (details in the
appendix).
\end{remark}
\begin{notheorembrackets}
\begin{lemma}[\cite{WDKH2019}]
When $\sqsubseteq$ is a partial order on $F$, pointed coalgebras and lax homomorphisms
form a category, which we denote by $\FSim(1,F)$.
\end{lemma}
\end{notheorembrackets}

Much as with open maps, many flavors of bisimilarity can be recovered using spans of 
proper homomorphisms:
\begin{definition}
We say that two pointed coalgebras are coalgebraically bisimilar if there is a span of 
proper homomorphisms between them.
\end{definition}

There are many ways of defining bisimilarities in coalgebra (see \cite{jacobs16} for an overview), 
but they coincide for the purpose of the present paper.

\section{Weighted Bisimilarity and Open Maps}
\label{sec:negativeResults}

In this section, we describe known attempts to model weighted systems, and 
particularly probabilistic ones, using open maps. They all work with some variations 
of the (discrete) distribution functor
on $\Set$. We will denote this functor, which maps a
set $X$ to the set
\[
	\Dist X = \big\{ f\colon X \to [0,1] \mid f^{-1}\big((0,1]\big)\text{ is finite and } \sum_{x\in X} f(x) = 1\big\},
\]
by $\Dist$ and the variation where the condition $= 1$ is replaced by $\leq 1$ by $\Subdist$
(i.e.~$\Subdist X := \Dist (X + 1)$).
We will prove that, even though Larsen-Skou bisimulations for reactive systems
can be modeled with open maps, that is impossible for bisimulations for 
generative systems.

\subsection{Larsen-Skou bisimilarity for reactive systems using open maps}

In \cite{cheng95}, Cheng et al. describe an open map situation for Probabilistic Transition Systems (PTSs), 
which corresponds to coalgebras for the functor $(\Dist(\arg)+1)^A$. In this setting, they 
consider Partial PTSs (PPTS) which are coalgebras for  
$(\Subdist^\varepsilon(\arg)+1)^A$ where 
the sub-probability distributions can have values in hyperreals, allowing infinitesimals $\varepsilon$. The category of PTSs 
embeds in that of PPTSs, and the path category is the full subcategory of PPTSs consisting of finite linear systems 
whose probabilities of transitions are infinitesimals. It is then proved that $J$-bisimilarity, restricted to PTSs,
for this path category 
corresponds to Larsen-Skou's probabilistic bisimilarity \cite{larsen91}.

This open map situation has been reformulated in \cite{dubut18} in terms of coreflections: the obvious functor from 
PPTSs to TSs is a coreflection whose left-adjoint maps a LTS $T$ to the PPTS whose 
underlying LTS is $T$ and where all transitions have infinitesimal probabilities. 
In general, given a coreflection $F:\,\mathcal{C}\,\to\,\mathcal{D}$ with left-adjoint $G$ and a 
path category $J$ on $\mathcal{D}$, one automatically has the path category $G\circ J$ on 
$\mathcal{C}$, and this construction preserves good properties of $J$. 
In particular, one has that two systems $A$ and $B$ are $(G\circ J)$-bisimilar if and only if 
$FA$ and $FB$ are $J$-bisimilar. Cheng et al.'s path category is obtained in this manner
with the coreflection above and the standard path category on LTSs. In particular, it means that two 
PPTSs are bisimilar if and only if their underlying TSs are strongly bisimilar.

\subsection{Impossibility result for generative systems}
\label{sec:impossibility}

In \cite{desharnais02}, Desharnais \etal~ describe several bisimilarities for 
generative probabilistic systems, that is, coalgebras for the functor
$\Subdist(A\times \arg)$, in a coalgebraic way. 
They pointed out that their efforts to model those bisimilarities using open maps
failed \cite[p.~188]{desharnais02}.
In the following, we see that it is in fact not possible. 
We will show that for generative probabilistic 
systems modeled by the category
$M:=\FSim(1,\Subdist(A\times \arg))$, there is no open map characterization 
of the coalgebraic bisimilarity.
Actually, the argument here is valid for many 
other types of weights and is not limited to reals.
Here, for two functions $f,g\colon X \to \Subdist(Y)$, $f \sqsubseteq g$ means 
that for all $x \in X$, for all $y \in Y$, $f(x)(y) \leq g(x)(y)$, where $\leq$ is 
the usual ordering on $[0,1]$.

In this situation:

\begin{theorem}
\label{thm:impossibility}
  For $\M:= \FSim(1,\Subdist(A\times \arg))$ there is no category $\bbP$ and no
  functor $J\colon \bbP\to \M$ such that for every $h\colon X\to Y$ with
  reachable $X$ the following equivalence holds:
  \[
    h\text{ is $J$-open}
    \quad
    \Longleftrightarrow
    \quad
    h\text{ is a proper homomorphism}
  \]
  and there is no $\bbP$ and no functor $J$ such that for every $X$ and $Y$:
  \[
  	X\text{ and }Y\text{ are }J\text{-bisimilar}
        \quad
        \Longleftrightarrow
        \quad
        X\text{ and }Y\text{ are coalgebraically bisimilar}.
   \]
\end{theorem}
\begin{proof}[Sketch]
By contradiction, assume that there is such a $J$.
We prove that there is a proper homomorphism of the form:
 \[
    X = \simplepath[30]{{1/n,a}/x_1,{}/{\raisebox{6pt}{\ensuremath{\vdots}}},{1/n,a}/x_n} \quad\quad\xrightarrow{~h~}\quad\quad
    Y = \simplepath{{1,a}/y}
\]
which cannot be $J$-open.
Consider first the unique lax homomorphism $0_\M \to Y$ where $0_\M$ consists in one state 
and no transition. This is not a proper homomorphism, so it is not open by assumption.
That is there is a square:
  \[
    \begin{tikzcd}
      JP
      \arrow{r}[description,inner sep=2pt]{p}
      \arrow{d}[description,inner sep=2pt]{J\phi}
      & 0_\M
      \arrow{d}[description,inner sep=2pt]{!_Y}
      \\
      JQ \arrow{r}[description,inner sep=2pt]{q}
      & Y
    \end{tikzcd}
  \]
with no lifting. It is mechanical to check that $JP \simeq 0_\M$ and $JQ$ has at least 
one transition from its initial state to another state $r \xtransto{~w,a~} z$ with $w \neq 0$.
With $n = 2\cdot \lceil \frac{1}{w}\rceil$, the proper homomorphism $h$ above is not open: 
there cannot be a morphism
from $JQ$ to $X$ because $w > \frac{1}{n}$.
\qed
\end{proof}

\section{Generalized Open Maps}
\label{sec:generalOpenMaps}

The main argument of the proof of impossibility is
the fact that sometimes, a transition with some probability $w$ in the codomain comes from
probabilities $w_1, \ldots, w_n$ with $\sum_i w_i = w$ in the domain, which makes a lifting morphism impossible
with the current framework of open maps.

In this section, we will extend the open map framework with the main intuition that the lifting morphism \emph{splits} the probability $w$ into smaller parts $w_1,\ldots, w_n$.
After defining these generalized open maps, we
show some basic properties of the bisimilarity generated by them.
 
 \subsection{Generalized Open Maps Situation}
 
 Here, we describe our extension of the open maps framework. 
 The data is similar: we start with a category of models $\M$, but we need more 
 than just a functor $J\colon \bbP\to\M$. Assume:
 \begin{itemize} 
 	\item a set $V$ together with a function $J\colon V \to \ob{\M}$,
	\item two small categories $\E$ and $\bbS$ whose sets of objects are $V$,
	\item two functors $\JE\colon \E\to\M$ and $\JS\colon \bbS\to\M$ coinciding with $J$ on objects.
\end{itemize}

The classical open maps situation $J\colon \bbP\to\M$ fits in this extension 
as follows. The category $\E$ is given by $\bbP$ with the intention that they 
model path shapes and their \emph{extensions}. The functor $\JE$ is 
given by $J$. The category $\bbS$ is given by the discrete category 
$\lvert\bbP\rvert$, that is, the category whose objects are those of $\bbP$ 
and whose morphisms are only identities. The functor $\JS$ is the only 
possible one respecting the conditions of the definition above.

In the general context of this extension, the interpretation is a bit different. 
Now $V$ is meant to be a set of trees labelled by alphabets and weights. 
$\E$ still consists in extensions, extending trees into trees with longer branches. 
$\bbS$ then consists in \emph{merging morphisms}, similar to the description above:
for the example of weighted systems, those morphisms are allowed to merge states
into one, as long as they sum up the weights of the in-going branches.
Generally, those morphisms are allowed to perform some merges that are 
harmless for bisimilarity.

With this data, we can define generalized open maps:

\begin{definition}
\label{def:generalized-open}
A morphism $f\colon X\to Y$ in $\M$ is called \emph{(\E,\bbS)-open} if it satisfies the following \emph{lifting property} for all $e\colon v\to w$ in $\E$:
    \[
        \text{for all }
        \quad
        \begin{tikzcd}
            Jv
            \arrow{r}[description,inner sep=2pt]{x}
            \arrow{d}[description,inner sep=2pt]{\JE e}
            \descto{dr}{\commutes}
            & X
            \arrow{d}[description,inner sep=2pt]{f}
            \\
            Jw
            \arrow{r}[description,inner sep=2pt]{y}
            & Y
        \end{tikzcd}
        \quad
        \text{ there is }
        \quad
        \begin{tikzcd}
            Jv
            \arrow{rr}[description,inner sep=2pt]{x}
            \arrow[dashed]{dr}[description,inner sep=2pt]{\JE e'}
            \arrow{dd}[description,inner sep=2pt]{\JE e}
            \descto[pos=0.45]{drr}{\commutes}
            \descto[pos=0.75]{ddrr}{\commutes}
            & & X
            \arrow{dd}[description,inner sep=2pt]{f}
            \\
            \phantom{}
            \descto[pos=0.5]{r}{\commutes}
            &Ju
            \arrow[dashed]{ur}[description,inner sep=2pt]{x'}
            \arrow[dashed]{dl}[description,inner sep=2pt]{\JS s}
            & \phantom{}
            \\
            Jw
            \arrow{rr}[description,inner sep=2pt]{y}
            & & Y
        \end{tikzcd}
    \]
\end{definition}
The interpretation starts the same as in usual open maps. 
Assume that we have a tree $y$ in $Y$ 
extending the image by $f$ of the tree $x$ in $X$. 
If $f$ is open, there should be a tree $x'$ extending $x$ and whose image 
by $f$ is $y$. However, $x'$ may have a different shape than $y$, since it might 
be necessary to split transitions. That is what $u$ and $s$ are modeling: 
$w$ is obtained from $u$ by merging some states.

The connection with the classical open maps can be formulated as follows
\begin{proposition}
Given a functor $J\colon \bbP\to\M$ and a morphism $f\colon X\to Y$,
\[
    f\text{ is }J\text{-open if and only if }f\text{ is }(\bbP,\lvert\bbP\rvert)\text{-open.}
\]
\end{proposition}

Again, bisimilarity can be defined as the existence of a span of open maps
\begin{definition}
We say that $X$ and $Y$ are $(\E,\bbS)$-bisimilar if there is a span of 
$(\E,\bbS)$-open maps between them.
\end{definition}

 \subsection{Basic Properties}
 \label{sec:basic-properties}
 
 In this section, we will prove general properties of
 $(\E,\bbS)$-bisimilarity similar to the classical case. First, we show that if
$\M$ has pullbacks, then $(\E,\bbS)$-bisimilarity is an equivalence relation.
Secondly, we describe two notions of path bisimulations, both implied by $(\E,\bbS)$-bisimilarity.
Finally, we prove that it is enough to check openness on some generators of $\E$.

 \medskip

\noindent\medskip%
\begin{minipage}{.65\linewidth}
 In order to see when $(\E,\bbS)$-bisimilarity is an equivalence relation, we need to check symmetry, reflexivity, and transitivity.
 \emph{Symmetry} always holds because we can always swap the legs of the span.
 For \emph{reflexivity}, it is enough to prove that identities are open
  which is valid because $\bbS$ is a category and $\JS$ is a functor, as shown in the diagram on the right.
  The proof of \emph{transitivity} relies on composition and pullbacks:
 \end{minipage}%
 \begin{minipage}{.35\linewidth}%
 \hspace*{0pt}\hfill
     \begin{tikzcd}
            Jv
            \arrow{rr}[description,inner sep=2pt]{x}
            \arrow[dashed]{dr}[description,inner sep=2pt]{\JE e}
            \arrow{dd}[description,inner sep=2pt]{\JE e}
            \descto[pos=0.45]{drr}{\commutes}
            \descto[pos=0.75]{ddrr}{\commutes}
            & & X
            \arrow{dd}[description,inner sep=2pt]{\id}
            \\
            \phantom{}
            \descto[pos=0.5]{r}{\commutes}
            &Jw
            \arrow[dashed]{ur}[description,inner sep=2pt]{y}
            \arrow[dashed]{dl}[description,inner sep=2pt]{\JS\id = \id}
            & \phantom{}
            \\
            Jw
            \arrow{rr}[description,inner sep=2pt]{y}
            & & Y
        \end{tikzcd}
 \end{minipage}
  
  \begin{lemma}
  \label{lem:open-composition}
  	$(\E,\bbS)$-open maps are closed under composition and pullbacks.
  \end{lemma}

  \begin{theorem}
  \label{thm:equivalenceRelation}
  If $\M$ has pullbacks, then $(\E,\bbS)$-bisimilarity is a transitive relation, 
  and thus is an equivalence relation.
  \end{theorem}

 \subsubsection{Generalized Path Bisimulations.}
 In the classical open map setup~\cite{joyal96},
 another notion of bisimilarity can be defined by using path extensions directly:
 so-called strong path and path 
 bisimulations, which can be generalized as follows.
 Like originally~\cite{joyal96}, we assume that there is an element $0 \in V$, such that
 $J0$ is an initial object of $\M$ (note that $0$ is not required to be initial in $\E$ or $\bbS$). The intuition is that the unique morphism $!_X\colon J0\to X$ points to 
 the initial state of $X$.
 For example, $J0$ can be given by $(1,\id_1,\bot)$ in a category of pointed coalgebras if $1$ is the final object of $\C$ and if $\C(1,F1)$ has the least element $\bot\colon 1\to F1$
  (those conditions hold in the cases of interest).
 
 \begin{definition}
 \label{def:path-simulation}
 A \emph{path simulation} from $A$ to $B$ in $\M$
 is a set $R$ of spans of the form
 $A \xleftarrow{~a~} Jv \xrightarrow{~b~} B$ (for $v\in V$)
 satisfying the following two properties
 \begin{itemize}
 	\item \textbf{initial condition:} the span $A \xleftarrow{~!_A~} J0 \xrightarrow{~!_B~} B$ belongs to $R$.
	\item%
  \begin{minipage}[t]{.65\linewidth}
  \textbf{forward closure:} for all spans $A \xleftarrow{~a~} Jv \xrightarrow{~b~} B$ in $R$,
	all $e\colon v\to w \in \E$ and all $a'\colon Jw\to A \in \M$ such that $a = a'\cdot \JE e$, there are
	$e'\colon v\to u \in \E$, $s\colon u\to w \in \bbS$, and $b'\colon Ju\to B \in \M$ such that
	$\JE e = \JS s\cdot \JE e'$, $b = b'\cdot \JE e'$, and the span 
	$A \xleftarrow{~a'\cdot \JS s~} Ju \xrightarrow{~b'~} B$ belongs to $R$.
  \end{minipage}%
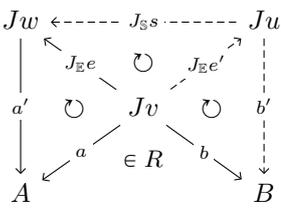
  \begin{minipage}[t]{.35\linewidth}%
  \hspace*{0pt}%
  \hfill%
	\begin{tikzcd}[baseline=(Jw.base)]
    |[alias=Jw]|
		Jw
		\arrow{dd}[description,inner sep=2pt]{a'}
		& \phantom{}
		\descto[pos=0.5]{d}{\commutes}
		& Ju
		\arrow[dashed]{ll}[description,inner sep=2pt]{\JS s}
		\arrow[dashed]{dd}[description,inner sep=2pt]{b'}
		\\
		\phantom{}
		\descto[pos=0.5]{r}{\commutes}
		& Jv
		\arrow{lu}[description,inner sep=2pt]{\JE e}
		\arrow[dashed]{ru}[description,inner sep=2pt]{\JE e'}
		\arrow{ld}[description,inner sep=2pt]{a}
		\arrow{rd}[description,inner sep=2pt]{b}
		\descto[pos=0.5]{r}{\commutes}
		\descto[pos=0.5]{d}{\ensuremath{\in R}\xspace}
		& 
    \phantom{}
    \\
		A
		&\phantom{}
		& B
       	\end{tikzcd}%
    \end{minipage}%
\end{itemize}
We say that $R$ is a \emph{strong path simulation} if it additionally satisfies the following:
\begin{itemize}
	\item\begin{minipage}[t]{.65\linewidth}
  \textbf{backward closure:} for all spans $A \xleftarrow{\,a\,} Jv \xrightarrow{\,b\,} B$ in $R$
	and all $e\colon w\to v \in \E$, we have that the span $A \xleftarrow{~a\cdot \JE e~} Jw \xrightarrow{~b\cdot \JE e~} B$
	belongs to $R$.
  \end{minipage}
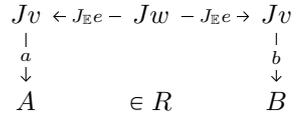
  \begin{minipage}[t]{.34\linewidth}
  \hspace*{0pt}\hfill%
	\begin{tikzcd}[baseline=(Jw.base)]
		Jv \arrow{d}[description]{a}
		& 
    |[alias=Jw]|
		Jw
		\arrow{l}[description,inner sep=2pt]{\JE e}
		\arrow{r}[description,inner sep=2pt]{\JE e}
		& Jv
    \arrow{d}[description]{b}
		\\
		A
		& \in R
		& B
  \end{tikzcd}
  \end{minipage}
 \end{itemize}
	We say that $R$ is a \emph{(strong) path bisimulation} from $A$ to $B$ if $R$ and 
	$R^\dagger = \{B \xleftarrow{~b~} Jv \xrightarrow{~a~} A \mid A \xleftarrow{~a~} Jv 
	\xrightarrow{~b~} B \in R\}$ are (strong) path simulations.
 \end{definition}
 
 Remark that this version of (strong) path bisimulations has the same type as
 the one by Joyal \etal~\cite{joyal96}, but satisfies more general conditions.
 In particular, when $\bbS$ is a discrete category, the formulation 
 above is exactly the one from \cite{joyal96}.
 Obviously, a strong path bisimulation is a path bisimulation.
 
 The main result of this section is the following.
 \begin{theorem}
 \label{thm:ES-implies-strong-path}
 Assume two models $A$ and $B$ in $\M$.
 If there is a span $A \xleftarrow{~f~} C \xrightarrow{~g~} B$ where $g$ is a morphism 
 of $\M$ and $f$ is an $(\E,\bbS)$-open map, then the following set is a strong path simulation:
 \[
  R_{f,g} := \set{
  A \xleftarrow{~a~} Jv \xrightarrow{~b~} B
  \mid
  \exists
  \text{$c\colon Jv\to C$ with $a = f\cdot c$ and $b = g\cdot c$}
  }
 \]
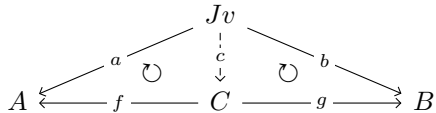
	\[
	\begin{tikzcd}
		\phantom{}
		&\phantom{}
		& Jv
		\arrow{lld}[description,inner sep=2pt]{a}
		\arrow{rrd}[description,inner sep=2pt]{b}
		\arrow[dashed]{d}[description,inner sep=2pt]{c}
		\descto[pos=0.6]{dl}{\commutes}
		\descto[pos=0.6]{dr}{\commutes}
		& \phantom{}
		& \phantom{}\\
		A
		&\phantom{}
		&C
		\arrow{ll}[description,inner sep=2pt]{f}
		\arrow{rr}[description,inner sep=2pt]{g}
		&\phantom{}
		& B
       	\end{tikzcd}
    	\]
 Consequently, if $A$ and $B$ are $(\E,\bbS)$-bisimilar, then there is 
 strong path bisimulation between them.
 \end{theorem}
 
 As in the classical case of \cite{joyal96}, there is no reason 
 for the converse to be true in general: there might be a 
 strong path bisimulation between two models, but no span 
 of generalized open maps.
 However, conditions from \cite{dubut16} could be 
 accommodated to describe a general framework in which the 
 converse holds. Since this is not the main focus of this paper, 
 we will not do it here, but will show a particular case in 
 Section~\ref{sec:weak-bisimulations}.
 
 \subsubsection{Generators of the Category of Extensions.}
 
 In the first example of open maps for LTSs introduced in Section~\ref{sec:pathCat}, 
 the path category was described as the poset of words with the prefix order.
 Consequently, to prove that a functional simulation is $J$-open,
 we have to prove the lifting property of Definition~\ref{def:generalized-open} with respect
 to all pairs $w \leq w'$. However, it is sufficient to check the
 lifting property for extensions by one letter: $w' = w.a$ for some $a \in A$. The general reason is 
 that, as a category, $(A^\ast,\leq)$ is generated by the morphisms $w \leq w.a$, 
 and verifying the lifting property with respect to generators of the category $\bbP$ is enough 
 to obtain $J$-openness.
 This can be extended to generalized open maps, with additional care.
 
 \begin{proposition}
 \label{prop:generators}
 Assume a subgraph $\E'$ of $\E$ that generates $\E$, that is,
 every morphism of $\E$ is a finite composition of morphisms 
 of $\E'$. Assume additionally, that for every $e \in \E'$ and $s \in \bbS$ for which
 $\JE e \cdot \JS s$ is well-defined, there are $s' \in \bbS$ and $e' \in \E'$ such that
 $\JE e \cdot \JS s = \JS s' \cdot \JE e'$.

 In that case, if a morphism of $\M$ satisfies the lifting property of
 Definition~\ref{def:generalized-open} for all morphisms in $\E'$,
 then it is $(\E,\bbS)$-open.
 Also, if a set of spans satisfies the conditions of Definition~\ref{def:path-simulation}, where $\E$ is replaced by $\E'$, 
 then it is a (strong) path bisimulation.
 \end{proposition}
 
 The first condition is satisfied when $\E$ is a free category and $\E'$ is its class of generators.
 The second condition is satisfied for e.g.\ $\E = \bbP$ and $\bbS = \lvert \bbP \rvert$.

\section{Open Maps for Weighted Systems}
\label{sec:open-weighted}

In this section, we will prove that weighted systems can be 
captured by this generalized open map theory for a large variety
of weights, including those needed to capture probabilistic systems.

\subsection{Category of Coalgebras for Weighted Systems}

In this section, we will consider weighted functors as follows.

\begin{definition}
Given a commutative monoid $(K, +, e)$, 
the K-weighted functor $\weighted{(K,+,e)}{\_}\colon\Set\to\Set$ is defined as follows
on sets and maps:
\begin{align*}
	\text{sets:}&& X 
		~~&\mapsto~~ \weighted{(K,+,e)}{X}=
			\set[\big]{\mu\colon X\to K \mid \mu^{-1}(K\setminus \{e\}) \text{ is finite}}\\
	\text{maps:}&& f\colon X\to Y 
		~&\mapsto~ \weighted{(K,+,e)}{f}(\mu) = 
			\big(y \in Y \mapsto \sum\set{\mu(x) \mid x\in X, f(x) = y}
        \big)
\end{align*}
\end{definition}
An element $\mu$ of $(K,+,e)^{(X)}$ is a finite distributions sending each
$x\in X$ to a weight in $K$. Whenever a map $f\colon X\to Y$ identifies elements $f(x_1) = f(x_2) = \cdots$, then
the functor action turns $\mu$ into a distribution on $Y$ by adding up the
weights $\mu(x_1) + \mu(x_2) + \cdots$ as elements of $X$ are sent to the same
element in $Y$. Since $\mu$ is finite and $K$ is commutative, this addition is
well-defined.

Given a commutative monoid $(K, +, e)$ and an alphabet $A$,
we want to consider weighted systems as 
coalgebras for the functor $\weighted{(K,+,e)}{A\times\_}$.
As described in Section~\ref{sec:pathCat}, we want to be able to talk about 
lax homomorphisms, so we need an order on 
$\weighted{(K,+,e)}{A\times\_}$  as in
Definition~\ref{def:partialOrder}.
For that, we need to assume an ordered commutative monoid
$(K,+,e,\sqsubseteq)$, that is, a monoid $(K, +, e)$ with a 
partial order $\sqsubseteq$ such that $+$ is monotone in both 
its arguments.

\begin{lemma}
\label{lem:orderMonoidFunctor}
Given an ordered commutative monoid $(K,+,e,\sqsubseteq)$, 
then for all sets $X$ and $Y$, the relation on
the hom-set
$\Set\big(X,\weighted{(K,+,e)}{A\times Y}\big)$ defined by
\[
	f_1 \sqsubseteq f_2 ~\Longleftrightarrow~ \forall x\in X,\, \forall y\in Y,\, \forall a \in A,\,
		f_1(x)(a,y) \sqsubseteq f_2(x)(a,y)
\]
is an order on $\weighted{(K,+,e)}{A\times\_}$.
\end{lemma}

So, we have a category $\FSim\left(1,\weighted{(K,+,e)}{A\times\_}\right)$ of 
pointed coalgebras and lax homomorphisms. The goal of this section
is to design a generalized open maps situation for which 
$(\E,\bbS)$-bisimilarity characterizes coalgebraic bisimilarity and more 
precisely for which $(\E,\bbS)$-openness characterizes 
proper homomorphisms.

In the course of the constructions and proofs, we will need additional 
assumptions that we list here.
\begin{definition}
We call an ordered commutative monoid $(K,+,e,\sqsubseteq)$ a \emph{rearrangement monoid}
if it satisfies the additional requirement that if $n,m\geq 1$ and
\[
	\sum\limits_{i=1}^n x_i  \sqsubseteq
		\sum\limits_{j=1}^m y_j,
\]
then there exists a family $(u_{i,j})_{1\leq i \leq n,1\leq j\leq m}$ such that
\[ \text{for all } j,\, \sum\limits_{i=1}^n u_{i,j} \sqsubseteq 
		y_j ~~\text{and}~~
		\text{for all } i,\, \sum\limits_{j=1}^m u_{i,j} = 
		x_i.
\]
In addition, we say that a rearrangement monoid is \emph{strict} 
if the condition above holds 
also when replacing $\sqsubseteq$ with $=$.
\end{definition}
The intuition is as follows. 
We have some weights arranged as $x_1, \ldots, x_n$.
We want to be able to decompose those weights into 
smaller weights, the $u_{i,j}$s, and by rearranging those 
small weights obtaining weights smaller than the $y_j$.
This condition states that this is possible when there is 
enough weight in total. 
The special case of strictness is called the \emph{row-column property} in \cite{klin09}.

\begin{lemma}
\label{lem:rearrangementMonoids}
For any subgroup $G$ of the real numbers $(\Real^n,+,-,0)$
such that
for all $x$, $y$ in $G$
$(\min(x_1,y_1), \ldots, \min(x_n,y_n)) \in G$,
the monoids
$(G,+,0,\leq)$ and $(G_{\geq 0},+,0,\leq)$, 
where $\leq$ is the usual order on $\Real^n$,
are strict rearrangement monoids.

For any lattice with bottom element 
$(L, \leq, \sqcup, \sqcap, \bot)$,
$(L,\sqcup, \bot, \leq)$ is a rearrangement monoid
if and only if $(L, \leq, \sqcup, \sqcap)$ is distributive.
Furthermore, in that case, it is always strict.

\end{lemma}

Another property is a form of positivity: we say that an ordered monoid is 
\emph{positively ordered}
if $e$ is the bottom element for $\sqsubseteq$, that is, for all $k \in K$, $e \sqsubseteq k$.

\begin{example}
The positive real line $(\Realp,+,0,\leq)$ is a positively ordered strict rearrangement monoid
 and it is necessary to define 
probabilistic systems. Another example is the monoid of natural numbers $(\Nat,+,0,\leq)$, which 
defines the bag functor. 
Finally, any distributive lattice with bottom element $(L,\sqcup,\bot,\leq)$,
typically powerset lattices $(\pow(X),\cup,\varnothing,\subseteq)$, is too.
On the contrary, $(\Real,+,0,\leq)$ and $(\mathbb{Z},+,0,\leq)$ are 
strict rearrangement monoids but are
not positively ordered. 
Conversely $(\Nat_{\geq 1},\times,1,\leq)$ is positively ordered but not a rearrangement monoid.
Indeed, it is impossible to rearrange the inequality $2\times 5 \leq 3\times 4$.
\end{example}

\subsection{Generalized Open Maps Situation for Weighted Systems}
\label{sec:genOpenWeighted}
Let $(K,+,e,\sqsubseteq)$ be a commutative ordered monoid.
Elements of $V_K$ are 
\begin{itemize}
	\item either words on $A\times (K\setminus\{e\})$, $w = (a_1,k_1), \ldots, (a_n,k_n)$,
	\item or triples $(w_1,b,w_2)$ of a word $w_1$ on $A\times (K\setminus\{e\})$, 
		a letter $b \in A$,
		and a non-empty word $w_2$ on $(K\setminus\{e\})$.
\end{itemize}
\noindent
The function $J_K$ maps
\begin{itemize}
	\item a word $w = (a_1,k_1), \ldots, (a_n,k_n)$ to the system
\[
	Jw =
   \weightedsystem[0]{
   \begin{scope}[x=15mm,y=5mm]
    \node[state] (s1) at (1,0) {1};
    \node[state] (s2) at (2,0) {\cdots};
    \node[state] (sn) at (3,0) {n};
    \draw[transition] (r) to node {\scriptstyle (a_1,k_1)} (s1);
    \draw[transition] (s1) to node {\scriptstyle (a_2,k_2)} (s2);
    \draw[transition] (s2) to node {\scriptstyle (a_n,k_n)} (sn);
   \end{scope}
   }
\]
that is, to the coalgebra $Jw\colon \{0, \ldots, n\}\to\weighted{K}{A\times\{0,\ldots,n\}}$ such that 
if $b = a_{i+1}$ and $j = i+1$ then $Jw(i)(b,j) = k_{i+1}$, else $= e$.
	\item a triple $(w_1,b,w_2)$ with $w_1 = (a_1,k_1), \ldots, (a_n,k_n)$
	and $w_2 = l_1, \ldots, l_m$ is mapped to the system
   \[J(w_1,b,w_2) = 
   \weightedsystem[0]{
   \begin{scope}[x=15mm,y=5mm]
    \node[state] (s1) at (1,0) {1};
    \node[state] (s2) at (2,0) {\cdots};
    \node[state] (sn) at (3,0) {n};
    \node[state] (first leaf) at (4.5,1) {(n+1,1)};
    \node[state] (last leaf) at (4.5,-1) {(n+1,m)};
    \draw[transition] (r) to node {\scriptstyle (a_1,k_1)} (s1);
    \draw[transition] (s1) to node {\scriptstyle (a_2,k_2)} (s2);
    \draw[transition] (s2) to node {\scriptstyle (a_n,k_n)} (sn);
    \draw[transition] (sn) to node[sloped,above] {\scriptstyle (b,l_1)} (first leaf);
    \draw[transition] (sn) to node[sloped,below] {\scriptstyle (b,l_m)} (last leaf);
    \node[rotate=90] at ($ (first leaf) !.5! (last leaf) $) {\ensuremath{\cdots}};
   \end{scope}
   }
  \]
  \jdnote{}
that is, $J(w_1,b,w_2)(n)(n+1,i) = (b,l_i)$.
\end{itemize}
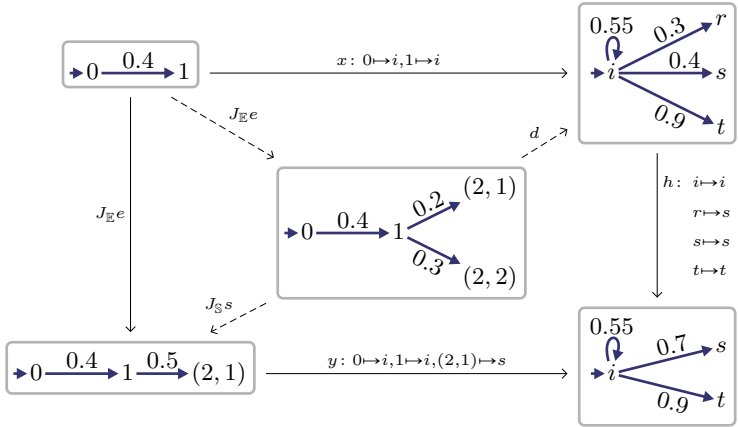
\begin{figure}[t]\centering
  \begin{tikzcd}[row sep=-4pt,column sep=0pt]
    \weightedsystem[0]{
      \node[state] (s) at (1,0) {1};
      \draw[transition] (r) to node{0.4} (s);
    }
    \arrow{rr}{x\colon 0\mapsto i, 1\mapsto i}
    \arrow{dd}[swap]{\JE e}
    \arrow[dashed]{dr}{\JE e}
    &&[4mm]
    \weightedsystem[i]{
      \node[state] (state r) at (1.2,0.6) {r};
      \node[state] (state s) at (1.2,0.0) {s};
      \node[state] (state t) at (1.2,-0.6) {t};
      \draw[transition] (r) to[loop above] node {0.55} (r);
      \draw[transition] (r) to node[pos=0.6,sloped,above] {0.3} (state r);
      \draw[transition] (r) to node[pos=0.7] {0.4} (state s);
      \draw[transition] (r) to node[pos=0.6,sloped,below] {0.9} (state t);
    }
    \arrow{dd}{h\colon\begin{array}[t]{l}
      i\mapsto i\\
      r\mapsto s\\
      s\mapsto s\\
      t\mapsto t
      \\
    \end{array}}
    \\[2mm]
    & \weightedsystem[0]{
      \node[state] (s) at (1,0) {1};
      \node[state] (t1) at (2,0.5) {(2,1)};
      \node[state] (t2) at (2,-0.5) {(2,2)};
      \draw[transition] (r) to node {0.4} (s);
      \draw[transition] (s) to node[sloped,above] {0.2} (t1);
      \draw[transition] (s) to node[sloped,below] {0.3} (t2);
    }
    \arrow[dashed]{dl}[swap]{\JS s}
    \arrow[dashed]{ur}{d}
    \\
    \weightedsystem[0]{
      \node[state] (s) at (1,0) {1};
      \node[state] (t) at (2,0) {(2,1)};
      \draw[transition] (r) to node{0.4} (s);
      \draw[transition] (s) to node{0.5} (t);
    }
    \arrow{rr}{y\colon 0\mapsto i, 1\mapsto i, (2,1)\mapsto s}
    &&
    \weightedsystem[i]{
      \node[state] (state s) at (1.2,0.3) {s};
      \node[state] (state t) at (1.2,-0.3) {t};
      \draw[transition] (r) to[loop above] node {0.55} (r);
      \draw[transition] (r) to node[pos=0.6,sloped,above] {0.7} (state s);
      \draw[transition] (r) to node[pos=0.6,sloped,below] {0.9} (state t);
    }
  \end{tikzcd}%
  \caption{Example of a lifting of a path extension in $\Realp$-weighted systems and for a singleton label alphabet $|A|=1$, thus omitting action labels.}
  \label{ex:weightedLift}
\end{figure}
  
The category $\E_K$ is defined as follows. For every $w_1$, $b$, and $w_2$, there is a unique 
edge $e$ from $w_1$ to $(w_1, b, w_2)$.
The functor then maps this edge $e$ to $\JE e$, the obvious injection.

The category $\bbS_K$ has two types of morphisms:
\begin{itemize}
	\item identities on words $w_1$,
	\item morphisms from $(w_1, b, w_2')$ to $(w_1, b, w_2)$, with $w_2' = l_1',\ldots,l_{m'}'$ and 
	$m \leq m'$, which are given by surjective monotone functions 
	$s\colon\{1,\ldots,m'\}\to\{1,\ldots,m\}$ such that for all $i \leq m$, 
	$
	l_i = \sum_{\{j \mid s(j) = i\}} l_j'.
	$
\end{itemize}
The functor $\JS$ then maps $s$ of the second type to the proper homomorphism 
$\JS s$ which maps $i$ to $i$ and $(n+1,j)$ to $(n+1,s(j))$.

As a piece of notation, for a morphism $x\colon Jw_1\to X$, with $w_1$ of length $n$
we denote $x(n) \in X$ by $\text{end}(x)$. We then say that a state $p$ of $X$ is 
reachable if there is a morphism of type $x\colon Jw_1\to X$ with 
$\text{end}(x) = p$. By extension, we say that $X$ is reachable if all its states are reachable.

\subsection{Equivalence between Open Maps and Proper Homomorphisms}
An example of an $(\E,\bbS)$-open map $h$ is provided in
\autoref{ex:weightedLift}, together with a path extension that is lifted. Like
it is often the case in the non-deterministic systems, the lifting map $d$ is not unique. Hence, only
existence (and no uniqueness) is required in the lifting property.
Since $h$ is a proper homomorphism, it provides a lifting for all extensions, as we show in general:
\begin{theorem}
\label{thm:weightedOpen}
Assume a lax homomorphism $f\colon X\to Y$.
If $f$ is $(\E_K,\bbS_K)$-open, $X$ is reachable, and $K$ is positively ordered, then 
$f$ is a proper homomorphism. Conversely, 
if $f$ is a proper homomorphism and $K$ is a rearrangement monoid, 
then $f$ is $(\E_K,\bbS_K)$-open.
In particular, if $K$ is a positively ordered rearrangement monoid, 
two weighted systems $X$ and $Y$ are $(\E_K,\bbS_K)$-bisimilar
if and only if they are coalgebraically bisimilar.
\end{theorem}

For an endofunctor on $\Set$, to prove that coalgebraic bisimilarity is an 
equivalence relation it is enough
to show that the functor preserves weak-pullbacks. In the case of the weighted
functor, this is given by strictness (see also \cite{klin09}):

\begin{corollary}
\label{cor:weightedTransitivity}
If $K$ is a positively ordered strict rearrangement monoid, 
then $(\E_K,\bbS_K)$-bisimilarity is an equivalence relation.
\end{corollary}

\subsection{About Sub-distribution Functor}

Until now, we have not dealt with probabilistic systems, that is, 
coalgebras for the sub-distribution functor $\Subdist$.
Those coalgebras are particular cases of coalgebras for the weighted functor
$X\mapsto\weighted{(\Real_+,+)}{X}$. 
We want to show in this section that it is equivalent to consider 
coalgebras for $X\mapsto\Subdist(A\times X)$ as coalgebras for 
$X\mapsto\weighted{(\Real_+,+)}{A\times X}$, in the sense that, two coalgebras
for the former are bisimilar if and only if they are bisimilar when seen as coalgebras 
for the latter.
The main ingredient is the following remark.

\begin{lemma}
\label{lem:from-R-to-D}
Assume a pointed coalgebra 
$1 \xrightarrow{~i~} X \xrightarrow{~c~} \Subdist(A\times X)$
and assume given a lax (resp. proper) homomorphism $f$ from
$1 \xrightarrow{~j~} Y \xrightarrow{~d~} \weighted{(\Real_+,+)}{A\times Y}$ to 
$1 \xrightarrow{~i~} X \xrightarrow{~c~} \Subdist(A\times X) \subseteq \weighted{(\Real_+,+)}{A\times X}$.
Then $Y \xrightarrow{~d~} \Subdist(A\times Y)$ and $f$ is a lax (resp. proper)
homomorphism from $1 \xrightarrow{~j~} Y \xrightarrow{~d~} \Subdist(A\times Y)$
to $1 \xrightarrow{~i~} X \xrightarrow{~c~} \Subdist(A\times X)$.
\end{lemma}

Remark that this property is not true for the proper distribution functor $\Dist$.
This suggests that we can define a generalized open maps situation 
$\E_{\Dist}, \bbS_{\Dist}$
for coalgebras for the functor $X\mapsto\Subdist(A\times X)$ by 
considering $\E_{(\Real_+,+)}, \bbS_{(\Real_+,+)}$
as defined in Section~\ref{sec:genOpenWeighted}, 
and restricting it to those $v$ such that 
$Jv$ is a coalgebra for $X\mapsto\Subdist(A\times X)$.

\begin{corollary}
\label{cor:D-open-iff-R-open}
A lax homomorphism from $1 \xrightarrow{~j~} Y \xrightarrow{~d~} \Subdist(A\times Y)$
to $1 \xrightarrow{~i~} X \xrightarrow{~c~} \Subdist(A\times X)$ 
is $(\E_{\Dist}, \bbS_{\Dist})$-open
if and only if it is $(\E_{(\Real_+,+)}, \bbS_{(\Real_+,+)})$-open.
Furthermore, two $\Subdist(A\times \_)$-coalgebras are $(\E_{\Dist}, \bbS_{\Dist})$-bisimilar 
if and only 
if they are $(\E_{(\Real_+,+)}, \bbS_{(\Real_+,+)})$-bisimilar.
\end{corollary}

Finally, the main result of this section:

\begin{theorem}
Let $f\colon X\to Y$ be a lax homomorphism between 
$\Subdist(A\times \_)$-coalgebras $(X,c,i)$ and $(Y,d,j)$.
If $(X,c,i)$ is reachable and $f$ is $(\E_{\Dist}, \bbS_{\Dist})$-open, then 
$f$ is a proper homomorphism.
Conversely, if $f$ is a proper homomorphism, then it is $(\E_{\Dist}, \bbS_{\Dist})$-open.
Moreover, two $\Subdist(A\times \_)$-coalgebras $(X,c,i)$ and $(Y,d,j)$ are
$(\E_{\Dist}, \bbS_{\Dist})$-bisimilar if and only if they are coalgebraically bisimilar.
\end{theorem}

\section{Open Maps for Branching Bisimilarity}
\label{sec:weak-bisimulations}

In this section, we present a new way of modeling branching and weak 
bisimulations using our generalized framework of open maps. 
Using this additional flexibility, we do not need to rely on weak 
morphisms anymore, but on a slight modification of the morphism 
described in Definition~\ref{def:lts}.
Concretely, we build a generalized open map situation such that 
stuttering branching bisimulations coincide with 
strong path bisimulations, and that in this case, they precisely 
characterize $(\E,\bbS)$-bisimilarity. In addition, in this framework, 
path bisimulations precisely correspond to weak bisimulations, 
witnessing branching bisimilarity as the history-preserving analogue 
to weak bisimilarity.

\subsection{LTSs with Internal Moves, Category and Bisimilarities}

\begin{definition}
For a fixed set $A$ of labels with a particular element $\tau$ 
(called \emph{internal move}), 
the category $\WLTS$ contains the same objects as $\LTS$, and its
morphisms $f\colon (X,\transto,x_0)\to (Y,\transto,y_0)$ 
are functions $f\colon X\to Y$ such that $f(x_0) = y_0$ and for all
$x\xtransto{a} x'$ in $X$, we have $f(x) \xtransto{a} f(x')$, 
or $a = \tau$ and $f(x) = f(x')$.
\end{definition}

$\LTS$ is a (non-full) subcategory of $\WLTS$, and in fact the $\LTS$-morphisms will be used later in the paper.
For easier distinction, we use the terminology \emph{strong morphisms} for
$\WLTS$-morphisms that are also in $\LTS$ (alluding to \emph{strong
bisimulations} which were the bisimulation notion in $\LTS$).
Another notion of morphisms are so-called \emph{weak morphisms}~\cite{cheng95}:
\begin{itemize}
	\item if $x\xtransto{~a~} x'$ in $X$, then 
		$f(x)\xtranstostar{~\tau~}\xtransto{~a~}\,\,\xtranstostar{~\tau~} f(x')$ in $Y$,
	\item if $x\xtransto{~\tau~} x'$ in $X$, then 
		$f(x) \xtranstostar{~\tau~}f(x')$ in $Y$.
\end{itemize}
Though we do not use weak morphisms in the following development of the paper,
it is worth mentioning the $\WLTS$-morphisms form a proper subclass of the weak
morphisms.

\begin{definition}
A \emph{branching bisimulation} from $(X,\transto_X,i_X)$ to $(Y,\transto_Y,i_Y)$ 
is a relation $R \subseteq X\times Y$ such that $(i_X,i_Y) \in R$, and for $(x,y) \in R$:
\begin{itemize}
	\item if $x \xtransto{~a~} x'$ then
		\begin{itemize}
			\item $a = \tau$ and $(x',y) \in R$, or
			\item $y \xtransto{~\tau~} y_1 \xtransto{~\tau~} \ldots \xtransto{~\tau~} y_n \xtransto{~a~} z_1 \xtransto{~\tau~} \ldots \xtransto{~\tau~} z_m$
			such that $(x,y_n)$, $(x',z_1)$, and $(x',z_m) \in R$.
		\end{itemize}
	\item symmetrically when $y \xtransto{~a~} y'$.
\end{itemize}
If furthermore in the second condition $(x,y_i),\,(x',z_i) \in R$ for all $i$ 
(and symmetrically in the third condition), then $R$ is said to be \emph{stuttering}.
\end{definition}

It is known from \cite{glabbeek96} that the largest branching bisimulation is stuttering, 
so that both notions generate the same bisimilarity. In the following, we will prove that 
strong path bisimulations are more naturally related to stuttering branching bisimulations
thanks to their backward closure.

\begin{definition}
A \emph{weak bisimulation} from $(X,\transto_X,i_X)$ to $(Y,\transto_Y,i_Y)$ 
is a relation $R \subseteq X\times Y$ such that $(i_X,i_Y) \in R$, and for $(x,y) \in R$:
\begin{itemize}
	\item if $x \xtransto{~\tau~} x'$, then there is $y'$ such that 
		$(x',y') \in R$ and $y \xtranstostar{~\tau~} y'$,
	\item if $x \xtransto{~a~} x'$ with $a \neq \tau$, 
		then there is $y'$ such that $(x',y') \in R$ and 
		$y \xtranstostar{~\tau~}\xtransto{~a~}\,\xtranstostar{~\tau~} y'$.
	\item symmetrically when $y \xtransto{~\tau~} y'$ or $y \xtransto{~a~} y'$.
\end{itemize}
\end{definition}
\noindent
It is clear that a (stuttering) branching bisimulation is a weak bisimulation.

\subsection{Generalized Open Maps for Branching Bisimulations}
\label{sec:path-weak}

In this section, we describe the generalized open maps situation
that captures branching bisimulation.
Like for plain LTSs (Def.~\ref{def:lts-paths}), elements of $V$ will be words 
on $A$, representing a finite linear LTS labelled by this word.
However, to emphasize the particularity of the internal move $\tau$, 
we will provide another presentation here.

Here, $V$ is the set of sequences of the form:
$v = n_1, a_1, n_2, \ldots, n_k, a_k, n_{k+1}$
such that $a_i \in A\setminus\{\tau\}$ and $n_i \in \Nat$, e.g.\ \(
  \tau \tau a \tau bc \tau \hateq 2,a,1,b,0,c,1
\).
The natural numbers $n_i\in \Nat \cong \set{\tau}^*$ represent the number of internal moves
between two observable moves.
Then, $J$ maps this sequence to the usual linear LTS:
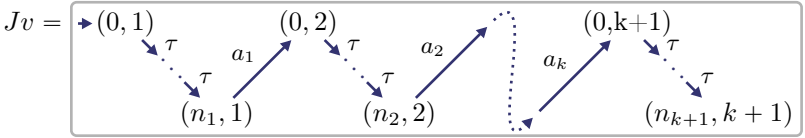
\[
  \newcommand{\tausequenceTo}[2]{%
    \path[draw=none] (#1) to node[sloped,text=transitioncolor,alias=dots] {\(\cdots\)} (#2);
    \draw[transition] (#1) to node[xshift=3pt,yshift=1pt,right] {\tau} (dots);
    \draw[transition] (dots) to node[above right] {\tau} (#2);
  }
  Jv = \weightedsystem[(0,1)]{
    \def\lowerRow{-1.0}
    \node[state] (n11) at (1,\lowerRow) {(n_1,1)};
    \node[state] (state02) at (2,0) {(0,2)};
    \node[state] (n22) at (3,\lowerRow) {(n_2,2)};
    \node (state03) at (4,0) {};
    \node (n33) at (4.5,\lowerRow) {};
    \node (state0k) at (5.5,0) {(0,k+1)};
    \node[state] (nkk) at (6.5,\lowerRow) {(n_{k+1}, k+1)};

    \tausequenceTo{r}{n11}

    \draw[transition] (n11) to node[above left] {a_1} (state02);
    \tausequenceTo{state02}{n22}

    \def\bezierLooseness{7mm}
    \draw[overlay,transition,dotted] (state03)
      .. controls ([xshift=\bezierLooseness,yshift=\bezierLooseness]state03)
         and ([xshift=-\bezierLooseness,yshift=-\bezierLooseness]n33)
      .. (n33);

    \draw[transition] (n22) to node[above left]{a_2} (state03);
    \draw[transition] (n33) to node[above left]{a_k} (state0k);
    \tausequenceTo{state0k}{nkk}
  }
\]
\noindent
Elements of $\E$ append \emph{at most} one observable (i.e.~non-$\tau$) move:
\begin{itemize}
	\item \textbf{Only internal moves}: for sequences 
		$v = n_1, a_1, \ldots, a_k, n_{k+1}$ and 
		$w = n_1, a_1, \ldots, a_k, n_{k+1}'$ with $n_{k+1} \leq n_{k+1}'$
		there is a unique edge $e_\tau\colon v\to w$ in $\E$, e.g.\
		\(
      e_\tau\colon
      2,a,1,b,0,c,1
      \to 2,a,1,b,0,c,3
    \)
	\item \textbf{One observable move}: for sequences 
		$v = n_1, a_1, \ldots, a_k, n_{k+1}$ and 
		$w = n_1, a_1, \ldots, a_k, n_{k+1}', a, n_{k+2}$ with $n_{k+1} \leq n_{k+1}'$
		there is a unique edge $e_a\colon v\to w$ in $\E$.
\end{itemize}

The graph morphism $\JE \colon\E\to\M$ maps those edges to the 
obvious inclusion, mapping state $(i,j)$ of $Jv$ to the same state in $Jw$.

Strictly speaking, $\E$ is not a category, but just a graph, because we have 
$a \xrightarrow{e_b} ab$ and $ab\xrightarrow{e_c} abc$, but there is no morphism 
from $a$ to $abc$.
To fit in the framework of 
Section~\ref{sec:generalOpenMaps}, we take the free category $\text{Free}(\E)$
generated by this graph and the 
unique functor extending the graph homomorphism $\JE$. 
By Proposition~\ref{prop:generators}, it is equivalent to consider $\text{Free}(\E)$ 
and $\E$ for 
openness and path bisimulations, so we will talk of 
$(\E,\bbS)$-openness, when we mean $(\text{Free}(\E),\bbS)$-openness, and all the 
statements and proofs will be done using $\E$ only.

Elements of $\bbS$ are trickier to describe. 
The intuition is that they are morphisms that merge states.
In the context of LTSs with internal moves, merging happens when the source 
and the target of a $\tau$-transition are mapped to the same state. This is crucial for the open 
maps we want to describe: to lift one $\tau$-transition, it might be necessary to 
use several $\tau$-transitions. 
With this knowledge, elements of $\bbS$ are as follows.
\begin{itemize}
	\item \textbf{Merging internal moves}:
		morphisms in $\bbS$ from 
		$v = n_1, a_1, \ldots, a_k, n_{k+1}$ to 
		$w = n_1', a_1, \ldots, a_k, n_{k+1}'$ with $n_i \geq n_i'$
		are $(k+1)$-tuples $s = (s_1, \ldots, s_{k+1})$ of monotone surjective functions 
		$s_i\colon\{0 < 1 < \ldots < n_i\}\to\{0 < 1 < \ldots < n_i'\}.$
\end{itemize}
For example, there are two morphisms from $a\tau\tau b \hateq 0,a,2,b,0$ to  $a\tau b \hateq 0,a,1,b,0$, one for each $\tau$ that can be dropped.
The functor $\JS$ then maps $s$ to the morphism from $Jv$ to $Jw$ defined by
$\JS(s)(i,j) = (s_j(i),j)$.

As a piece of notation, for a morphism $x\colon J(n_1, a_1, \ldots, a_k, n_{k+1})\to X$, 
we denote $x(n_{k+1},k+1) \in X$ by $\text{end}(x)$.

\subsection{Equivalence of Bisimilarities}

In this section, we prove that $(\E,\bbS)$-bisimilarity indeed coincides with branching bisimilarity.
To do so, we prove first that for the present instance of $\E$ and $\bbS$ (Sec.~\ref{sec:path-weak}),
$(\E,\bbS)$-bisimilarity coincides with strong path bisimilarity.
In general, $(\E,\bbS)$-bisimilarity implies strong path bisimilarity
(\autoref{thm:ES-implies-strong-path}), so it remains to show the converse
direction for the present instance.
To this end, we start by internalizing strong path bisimulations into objects
of $\LTS$/$\WLTS$, in order to relate it them to open maps:
\begin{definition} \label{def:strongPathLTS}
  For a strong path bisimulation $R$ from $X$ to $Y$,
  define the LTS $\widetilde{R} = (R, \transto_R, (X \xleftarrow{~!~} J0
  \xrightarrow{~!~} Y))$ to have transitions
\[ 
	(X \xleftarrow{x} Jv \xrightarrow{y} Y) \xtransto{~a~}_R (X \xleftarrow{x'} Jw \xrightarrow{y'} Y)
\]
\begin{itemize}
	\item for $a \neq \tau$ with $v = (n_1, a_1, \ldots, a_k, n_{k+1})$, 
		$w = (n_0, a_1, \ldots, a_k, n_{k+1}, a, 0)$,
		$x' = x\cdot \JE e_a$, and $y' = y\cdot \JE e_a$ (for the unique $e_a\colon v\to w$);
	\item for $a = \tau$ with $v = (n_1, a_1, \ldots, a_k, n_{k+1})$, 
		$w = (n_1, a_1, \ldots, a_k, n_{k+1}+1)$, 
		$x' = x\cdot \JE e_\tau$, and $y' = y\cdot \JE e_\tau$ (for the unique $e_\tau\colon v\to w$).
\end{itemize}
\end{definition}
As a first observation, we describe runs in $\widetilde{R}$ in terms of projection maps:
\begin{lemma}
\label{lem:run-R}
In \WLTS, we have projection maps
$X\xleftarrow{\pi_X} \widetilde{R}\xrightarrow{\pi_Y} Y$ given by
$
	\pi_X\colon (X \xleftarrow{x} Jv \xrightarrow{y} Y) 
		\mapsto \text{end}(x)$ and $
	\pi_Y\colon (X \xleftarrow{x} Jv \xrightarrow{y} Y) 
		\mapsto \text{end}(y)
$.
For every strong morphism $r\colon Jv\to\widetilde{R}$ (i.e.~$r\in \LTS$),
\[
	\runend(r)\text{ is of the form }(X\xleftarrow{~\pi_X\cdot r~}Jv\xrightarrow{~\pi_Y\cdot r~}Y).
\]
\end{lemma}

Remark that in this statement, we require $r$ to be strong and
not just a morphism of $\WLTS$. With a morphism of $\WLTS$, the statement
would become that there is $s\colon v'\to v \in \bbS$ such that $\pi_X\cdot r =
x\cdot \JS s$ instead. For the characterization of open maps in $\WLTS$, it
suffices for our needs to restrict to strong morphisms:

\begin{lemma}
\label{lem:restriction-to-strong}
For $f\colon X\to Y$ in $\WLTS$ to be $(\E,\bbS)$-open,
it is sufficient to verify the lifting in Definition~\ref{def:generalized-open}
in the special case of $x$ being a strong morphism.
\end{lemma}
\noindent
We use this simplification to prove that the projection maps $\pi_X,\pi_Y$ are open:

\begin{proposition}
\label{lem:path-implies-ES}
For a strong path bisimulation $R$ from $X$ to $Y$, the projections
$X\xleftarrow{\pi_X} \widetilde{R}\xrightarrow{\pi_Y} Y$
are both $(\E,\bbS)$-open.
\end{proposition}

The next step is to prove the equivalence between strong path and 
stuttering branching bisimulations.
\begin{theorem}
\label{lem:branching-stuttering-iff-strong-path}
If $R$ is a stuttering branching bisimulation 
from $X$ to $Y$, then
\[
	\overline{R} = \{ X \xleftarrow{x} Jv \xrightarrow{y} Y 
		\mid v = (n_1, a_1, \ldots, n_{k+1}) \wedge \forall i,\, j.\, (x(i,j),y(i,j)) \in R\}
\]
is a strong path bisimulation.
Conversely, if $R$ is a strong path bisimulation, then 
\[
	\widecheck{R} = \{ (\text{end}(x),\text{end}(y)) 
		\mid  (X \xleftarrow{x} Jv \xrightarrow{y} Y) \in R\}
\]
is a stuttering branching bisimulation.
\end{theorem}

The same reasoning can be made for weak and path bisimulations:
\begin{theorem}
\label{lem:weak-iff-path}
If $R$ is a weak bisimulation from $X$ to $Y$, then
\[
	\widehat{R} = \{ X \xleftarrow{x} Jv \xrightarrow{y} Y 
		\mid (\text{end}(x),\text{end}(y)) \in R\}
\]
is a path bisimulation. 
If $R$ is a path bisimulation, then 
$\widecheck{R}$
is a weak bisimulation.
\end{theorem}

In total, we can describe branching and weak bisimilarity by categorical bisimilarity notions, as summarized in \autoref{tab:equivalences}.

\begin{table}[t]
  \caption{Equivalences of bisimilarity notions in LTSs with $\tau$-actions $X,Y\in \WLTS$}
  \label{tab:equivalences}
  \centering
  \begin{tabular}{l@{~\(\Longleftrightarrow\)~}l@{~}l@{}}
    \toprule
    branching bisimilarity
    & strong path bisimilarity
    & (\autoref{lem:branching-stuttering-iff-strong-path})
    \\
    & $(\E,\bbS)$-bisimilarity
    & (\autoref{lem:path-implies-ES} \& \autoref{thm:ES-implies-strong-path})
    \\
    \midrule
    weak bisimilarity
    & path bisimilarity
    & (\autoref{lem:weak-iff-path})
    \\
    \bottomrule
  \end{tabular}
\end{table}

\section{Conclusions and Future Work}
\label{sec:conclusion}

In this paper, we investigate bisimilarities of weighted and probabilistic
systems through the theory of open maps. After showing that the usual theory
cannot capture weights, we provide a faithful extension of the theory by the
notion of mergings. The new theory has similar properties (equivalence relation, characterization as sets
of spans, restriction to generators) as classical open maps but also captures bisimilarity of weighted
systems and even branching bisimilarity.

The new instances come at the cost of more parameters to the theory. It remains
for future work whether the parameters $\E$, $\bbS$ can be combined in a
single path category with two morphism classes and morphism factorizations.
It would also be illuminating to know whether this new theory satisfies the 
axioms of a \emph{class of open maps} from \cite{joyal94}, in 
particular for toposes of coalgebras \cite{johnstone01}.

For the framework as presented, we would like to formally relate it to coalgebra -- 
as this has been done 
for
non-deterministic systems~\cite{lasota02,WDKH2019}. 
Furthermore, we would like to investigate how
system semantics of true concurrency, such as Higher Dimensional Automata
\cite{pratt05} can be integrated. Designing open maps for them turned out to be
complicated (see \cite{fahrenberg13}), but a hope would be that the addition of
mergings would allow modeling homotopy more naturally. 

Finally, it would be interesting to see whether our theory capture quantitative 
extensions of systems classically modeled by open maps, such as probabilistic 
and quantum extensions of petri nets and event structures 
(see \cite{winskel13} for example).

\label{maintextend} %
\newpage

\bibliographystyle{splncs04}
\bibliography{refs}

\newpage

\appendix

\section{Details for Section~\ref{sec:pathCat}}

\begin{proofappx}[Details for]{remOrderHJ}
  Hughes and Jacobs \cite{HughesJ04} work with preorders instead of partial orders, but adapting their definition~\cite[Def.~2.1]{HughesJ04} to partial orders yields:
  An order on a Set-functor $F\colon \Set\to\Set$ is a functor $\le\colon \Set\to \Pos$ making the following diagram commute:
  \[
    \begin{tikzcd}
    & \Pos \arrow{d}{U}
    \\
    \Set \arrow{ur}{\le} \arrow{r}{F}
    & \Set
    \end{tikzcd}
  \]
  Every order on a functor $\sqsubseteq$ in our sense (\autoref{def:partialOrder}) for $\C:=\Set$ induces such a functor
  $\le\colon \Set\to \Pos$, because for every set $Y$, we can put
  $\le_Y\in \Pos$ to be the order $\sqsubseteq$ on $\Set(1,FY)$, identifying
  elements of $FY$ with maps $1\to FY$.

  Conversely, given $\le\colon \Set\to\Pos$ in the sense of Hughes and Jacobs,
  we can define the order on $f_1,f_2\in \Set(X,FY)$ by
  \begin{equation}
    f_1 \sqsubseteq f_2 :\Longleftrightarrow \forall x\in X\colon f_1(x) \le_Y f_2(x).
    \label{eqPointwise}
  \end{equation}

  However, it might be possible that not all orders $\sqsubseteq$ on functors
  in the sense of \autoref{def:partialOrder} satisfy \eqref{eqPointwise}.
  \qed
\end{proofappx}

\section{Proof of Section~\ref{sec:negativeResults}}

Recall that an object of the category $\FSim(1,\Subdist(A\times\_))$ 
is a pointed coalgebra, that is, a pair of 
functions of the form
\[
	1 \xrightarrow{~i~} X \xrightarrow{~\xi~} \Subdist(A\times X).
\]
In the following, we will identify $i$ with the element $x_0 \in X$, such 
that the image of $i$ is $\{x_0\}$, and we will denote the pointed 
coalgebra above by $(X,\xi,x_0)$ instead.

\begin{proofappx}{thm:impossibility}
  Assume some $J$ with the equivalence, and note that wlog we can assume that
  $JP$ is reachable.\footnote{if there was such a $J$, then $R\cdot J$ is also
    such a candidate, where $R\colon \M\to \M_{\mathsf{reach}}$ restrict a pointed
    coalgebra to its reachable states.} For $Y= \simplepath{{1,a}/y}$ (implicitly
  defining $\zeta\colon Y\to \Subdist(A\times Y)$ by $\zeta(r)(a,y) = 1$), consider the
  initial morphism $!_Y\colon 0_\M\to Y$, where $0_\M =\simplepath{}$. 
  Here $0_\M$ has only one state and no transition.
  Since
  $!_Y$ is not a proper coalgebra homomorphism, it is not a $J$-open morphism.
  Hence there is some commutative square
  \[
    \begin{tikzcd}
      JP
      \arrow{r}[description,inner sep=2pt]{p}
      \arrow{d}[description,inner sep=2pt]{J\phi}
      & 0_\M
      \arrow{d}[description,inner sep=2pt]{!_Y}
      \\
      JQ \arrow{r}[description,inner sep=2pt]{q}
      & Y
    \end{tikzcd}
  \]
  for which no diagonal lifting exists. Since no diagonal exists, $J\phi $ is no
  isomorphism. Since $p$ is a lax coalgebra homomorphism, $JP$ has no edges and
  hence $JP\cong 0_\M$, $p$ is an iso. This implies that $JQ$ has at least one edge,
  because $JQ\not\cong JP$, and it must be between distinct elements and
  labelled with $a$ because $q$
  is a lax coalgebra homomorphism. Hence, there are $r,z \in JQ$,  with $r\neq z$
  and $r\xrightarrow{w,a}z$, $0< w\le 1$. Define $n = 2\cdot \lceil \frac{1}{w}
  \rceil$ and
  \[
    X = \simplepath{{1/n,a}/x_1,{}/{.},{1/n,a}/x_n},\text{ i.e. }
    X =\{r,x_1,\ldots,x_n\},~
    \xi(r)((a,x_k)) = \frac{1}{n}, 1\le k\le n.
  \]
  We have a proper coalgebra homomorphism $h\colon X\to Y$ given by $h(r)= r$
  and $h(x_k) = y$, $1\le k\le n$. It is a proper coalgebra homomorphism,
  because $\zeta(r)(y) = 1 = \sum_{1\le k \le n} \frac{1}{n} = \sum_{1\le k \le
    n} \xi(r)(x_k)$. By initiality of $0_\M$ and so of $JP$, the following diagram
  commutes:
  \[
    \begin{tikzcd}
      JP
      \arrow{r}[description,inner sep=2pt]{!_X}
      \arrow{d}[description,inner sep=2pt]{J\phi}
      & X
      \arrow{d}[description,inner sep=2pt]{h}
      \\
      JQ \arrow{r}[description,inner sep=2pt]{q}
      & Y
    \end{tikzcd}
  \]
  By the assumed equivalence, $h$ is $J$-open, and thus, there is some
  diagonal lifting $d\colon JQ\to X$. Now, $d$ being a lax coalgebra
  homomorphism leads to a contradiction, because $d$ maps the edge
  $r\xtransto{w,a} z$ in $JQ$ to some $d(r) \xtransto{w',a} d(z)$ with
  $w'\ge w$ in $X$. However, there is no such edge, because every edge in $X$
  has weight $\frac{1}{n}$ which is less than $w'$, which prove the first impossibility.
Now, since $h$ is a proper coalgebra homomorphism, $X$ and $Y$ are 
coalgebraically bisimilar. Let us prove they are not $J$-bisimilar. 
Assume that there is a span $Y \xleftarrow{~f~} Z \xrightarrow{~g~} X$\twnote{} of $J$-open 
morphisms. By a similar argument as earlier, $f$ being open means that there 
is a transition $r' \xtransto{w',a} z'$ in $Z$ with $w' \geq w$. Since $g$ is a 
lax homomorphism, $g$ will map this transition on a transition of the form 
$r \xtransto{w'',a} x_i$ with $w'' \geq w' \geq w$, which is again impossible.\qed
\end{proofappx}

\section{Proofs of Section~\ref{sec:generalOpenMaps}}

 \begin{proofappx}{lem:open-composition}
  Let us prove that they are closed under composition first.
  So start with the following situation
  \[
      \begin{tikzcd}
            Jv
            \arrow{r}[description,inner sep=2pt]{x}
            \arrow{dd}[description,inner sep=2pt]{\JE e}
            & X
            \arrow{d}[description,inner sep=2pt]{f}
            \\
            \phantom{}
            & Y
            \arrow{d}[description,inner sep=2pt]{g}
            \\
            Jw
            \arrow{r}[description,inner sep=2pt]{z}
            & Z
        \end{tikzcd}
  \]
  reorganizing the arrows, we obtain the following square on which we apply the openness of $g$
   \[
     \begin{tikzcd}
            Jv
            \arrow{rr}[description,inner sep=2pt]{f\cdot x}
            \arrow[dashed]{dr}[description,inner sep=2pt]{\JE e'}
            \arrow{dd}[description,inner sep=2pt]{\JE e}
            \descto[pos=0.45]{drr}{\commutes}
            \descto[pos=0.75]{ddrr}{\commutes}
            & & Y
            \arrow{dd}[description,inner sep=2pt]{g}
            \\
            \phantom{}
            \descto[pos=0.5]{r}{\commutes}
            &Ju
            \arrow[dashed]{ur}[description,inner sep=2pt]{d}
            \arrow[dashed]{dl}[description,inner sep=2pt]{\JS s}
            & \phantom{}
            \\
            Jw
            \arrow{rr}[description,inner sep=2pt]{z}
            & & Z
        \end{tikzcd}
  \]
  We then transform the top triangle into a square on which we apply the openness of $f$
   \[
     \begin{tikzcd}
            Jv
            \arrow{rr}[description,inner sep=2pt]{x}
            \arrow[dashed]{dr}[description,inner sep=2pt]{\JE e''}
            \arrow{dd}[description,inner sep=2pt]{\JE e'}
            \descto[pos=0.45]{drr}{\commutes}
            \descto[pos=0.75]{ddrr}{\commutes}
            & & X
            \arrow{dd}[description,inner sep=2pt]{f}
            \\
            \phantom{}
            \descto[pos=0.5]{r}{\commutes}
            &Ju'
            \arrow[dashed]{ur}[description,inner sep=2pt]{d'}
            \arrow[dashed]{dl}[description,inner sep=2pt]{\JS s'}
            & \phantom{}
            \\
            Ju
            \arrow{rr}[description,inner sep=2pt]{d}
            & & Y
        \end{tikzcd}
  \]
  Putting things together, we obtain the desired lifting
   \[
     \begin{tikzcd}
            Jv
            \arrow{rr}[description,inner sep=2pt]{x}
            \arrow[dashed]{dr}[description,inner sep=2pt]{\JE e''}
            \arrow[dashed]{dddr}[description,inner sep=2pt]{\JE e'}
            \arrow{dddd}[description,inner sep=2pt]{\JE e}
            \descto[pos=0.45]{drr}{\commutes}
            & & X
            \arrow{dd}[description,inner sep=2pt]{f}
            \\
            \phantom{}
            \descto[pos=0.4]{drr}{\commutes}
            &Ju'
            \arrow[dashed]{ur}[description,inner sep=2pt]{d'}
            \arrow[dashed]{dd}[description,inner sep=2pt]{\JS s'}
            \descto[pos=0.5]{dr}{\commutes}
            & \phantom{}
            \\
            \phantom{}
            \descto[pos=0.4]{dr}{\commutes}
            & \phantom{}
            & Y
            \arrow{dd}[description,inner sep=2pt]{g}
            \\
            \phantom{}
            &
            Ju
            \arrow[dashed]{ru}[description,inner sep=2pt]{d}
            \arrow[dashed]{dl}[description,inner sep=2pt]{\JS s}
            \descto[pos=0.4]{dr}{\commutes}
            & \phantom{}\\
            Jw
            \arrow{rr}[description,inner sep=2pt]{z}
            & \phantom{}
            & Z
        \end{tikzcd}
        \quad
        \text{that is}
        \quad
        \begin{tikzcd}
            Jv
            \arrow{rr}[description,inner sep=2pt]{x}
            \arrow[dashed]{dr}[description,inner sep=2pt]{\JE e''}
            \arrow{dd}[description,inner sep=2pt]{\JE e}
            \descto[pos=0.45]{drr}{\commutes}
            \descto[pos=0.75]{ddrr}{\commutes}
            & & X
            \arrow{dd}[description,inner sep=2pt]{g\cdot f}
            \\
            \phantom{}
            \descto[pos=0.5]{r}{\commutes}
            &Ju'
            \arrow[dashed]{ur}[description,inner sep=2pt]{d'}
            \arrow[dashed]{dl}[description,inner sep=2pt]{\JS(s\cdot s')}
            & \phantom{}
            \\
            Jw
            \arrow{rr}[description,inner sep=2pt]{z}
            & & Z
        \end{tikzcd}
  \]
  Now, let us prove they are closed under pullback, that is, if the following
	\[
	\begin{tikzcd}
            D
            \pullbackangle[dashed]{-45}
            \arrow{r}[description,inner sep=2pt]{k}
            \arrow{d}[description,inner sep=2pt]{h}
            \descto{dr}{\commutes}
            & C
            \arrow{d}[description,inner sep=2pt]{g}
            \\
            B
            \arrow{r}[description,inner sep=2pt]{f}
            & A
        \end{tikzcd}
	\]
	is a pullback square and $f$ is $(\E,\bbS)$-open, then $k$ is $(\E,\bbS)$-open.
  Start with the following situation:
  \[
  \begin{tikzcd}
            Jv
            \arrow{r}[description,inner sep=2pt]{d}
            \arrow{d}[description,inner sep=2pt]{\JE e}
            \descto{dr}{\commutes}
            & D
            \arrow{d}[description,inner sep=2pt]{k}
            \\
            Jw
            \arrow{r}[description,inner sep=2pt]{c}
            & C
        \end{tikzcd}
  \]
  Composing with the pullback square, we obtain the following on which we 
  apply the openness of $f$
  \[
  \begin{tikzcd}
            Jv
            \arrow{rr}[description,inner sep=2pt]{h\cdot d}
            \arrow[dashed]{dr}[description,inner sep=2pt]{\JE e'}
            \arrow{dd}[description,inner sep=2pt]{\JE e}
            \descto[pos=0.45]{drr}{\commutes}
            \descto[pos=0.75]{ddrr}{\commutes}
            & & B
            \arrow{dd}[description,inner sep=2pt]{f}
            \\
            \phantom{}
            \descto[pos=0.5]{r}{\commutes}
            &Ju
            \arrow[dashed]{ur}[description,inner sep=2pt]{b}
            \arrow[dashed]{dl}[description,inner sep=2pt]{\JS s}
            & \phantom{}
            \\
            Jw
            \arrow{rr}[description,inner sep=2pt]{g\cdot c}
            & & A
        \end{tikzcd}
  \]
  Reorganizing the arrows of the bottom right corner, we obtain the following square
   \[
  \begin{tikzcd}
            Ju
            \arrow{rr}[description,inner sep=2pt]{c\cdot \JS s}
            \arrow{d}[description,inner sep=2pt]{b}
            & \phantom{}
            \descto{d}{\commutes}
            & C
            \arrow{d}[description,inner sep=2pt]{g}
            \\
            B
            \arrow{rr}[description,inner sep=2pt]{f}
            & \phantom{}
            & A
        \end{tikzcd}
  \]
  Applying the universality of the pullback in the assumption, 
  we obtain the existence of a morphism $d'\colon Ju\to D$
  such that
  \[
  	k\cdot d' = c\cdot \JS s\quad\text{and}\quad h\cdot d' = b.
  \]
  Let us prove that the following lifting works
   \[
  \begin{tikzcd}
            Jv
            \arrow{rr}[description,inner sep=2pt]{d}
            \arrow[dashed]{dr}[description,inner sep=2pt]{\JE e'}
            \arrow{dd}[description,inner sep=2pt]{\JE e}
            \descto[pos=0.45]{drr}{\commutes}
            \descto[pos=0.75]{ddrr}{\commutes}
            & & D
            \arrow{dd}[description,inner sep=2pt]{k}
            \\
            \phantom{}
            \descto[pos=0.5]{r}{\commutes}
            &Ju
            \arrow[dashed]{ur}[description,inner sep=2pt]{d'}
            \arrow[dashed]{dl}[description,inner sep=2pt]{\JS s}
            & \phantom{}
            \\
            Jw
            \arrow{rr}[description,inner sep=2pt]{c}
            & & C
        \end{tikzcd}
  \]
  The bottom right corner ($k\cdot d' = c\cdot \JS s$) 
  and the left triangle ($\JS s\cdot \JE e' = \JE e$) are already given.
  Let us prove $d'\cdot \JE e' = d$.
  We already know that
  \[
  	h\cdot d'\cdot \JE e' = b \cdot \JE e' = h\cdot d.
  \]
  and that
  \[
  	k\cdot d'\cdot \JE e' = c \cdot \JS s\cdot \JE e' = c\cdot \JE e = k\cdot d.
  \]
  Then, by universality of the pullback square, $d'\cdot \JE e' = d$.
  \qed
  \end{proofappx}
  
   \begin{proofappx}{thm:equivalenceRelation}
  To finish with transitivity, assume given two spans of open maps as follows (in plain)
  
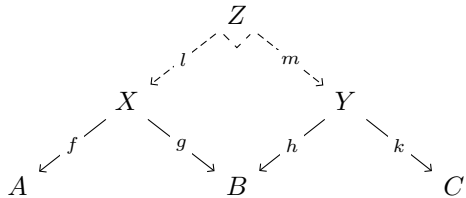
\[
  \begin{tikzcd}
  \phantom{}
  & \phantom{}
  & Z
  \pullbackangle[dashed]{-90}
  \arrow[dashed]{dl}[description,inner sep=2pt]{l}
  \arrow[dashed]{dr}[description,inner sep=2pt]{m}
  & \phantom{}
  & \phantom{}
  \\
  \phantom{}
  & X
  \arrow{dl}[description,inner sep=2pt]{f}
  \arrow{dr}[description,inner sep=2pt]{g}
  & \phantom{}
  & Y
  \arrow{dl}[description,inner sep=2pt]{h}
  \arrow{dr}[description,inner sep=2pt]{k}
  & \phantom{}
  \\
  A
  & \phantom{}
  & B
  & \phantom{}
  & C
  \end{tikzcd}
\]
  Forming the pullback of $g$ along $h$ (dashed), which are both open, and using Lemma~\ref{lem:open-composition}, 
  $l$ and $m$ are also open. Using Lemma~\ref{lem:open-composition} again, $f\cdot l$ and $k\cdot m$ are open, 
  and we have a span of open maps between $A$ and $C$.
  \qed
  \end{proofappx}

  \begin{proofappx}{thm:ES-implies-strong-path}
 Let us prove the three properties of the definition.
 \begin{itemize}
 	\item \textbf{initial condition:} By the universal property of the initial object $J0$
	the following diagram commutes:
	
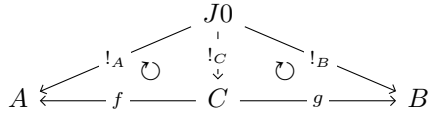
\[
	\begin{tikzcd}
		\phantom{}
		&\phantom{}
		& J0
		\arrow{lld}[description,inner sep=2pt]{!_A}
		\arrow{rrd}[description,inner sep=2pt]{!_B}
		\arrow[dashed]{d}[description,inner sep=2pt]{!_C}
		\descto[pos=0.6]{dl}{\commutes}
		\descto[pos=0.6]{dr}{\commutes}
		& \phantom{}
		& \phantom{}\\
		A
		&\phantom{}
		&C
		\arrow{ll}[description,inner sep=2pt]{f}
		\arrow{rr}[description,inner sep=2pt]{g}
		&\phantom{}
		& B
       	\end{tikzcd}
    	\]
	\item \textbf{forward closure:} Assume given the following situation:
	
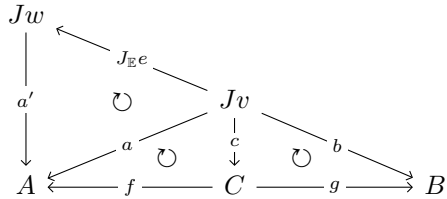
\[
	\begin{tikzcd}
		Jw
		\arrow{dd}[description,inner sep=2pt]{a'}
		& \phantom{}
		& \phantom{}
		& \phantom{}
		& \phantom{}
		\\
		\phantom{}
		\descto[pos=0.5]{rr}{\commutes}
		& \phantom{}
		& Jv
		\arrow{llu}[description,inner sep=2pt]{\JE e}
		\arrow{lld}[description,inner sep=2pt]{a}
		\arrow{rrd}[description,inner sep=2pt]{b}
		\arrow{d}[description,inner sep=2pt]{c}
		\descto[pos=0.6]{dl}{\commutes}
		\descto[pos=0.6]{dr}{\commutes}
		& \phantom{}
		& \phantom{}
		\\
		A
		& \phantom{}
		& C
		\arrow{ll}[description,inner sep=2pt]{f}
		\arrow{rr}[description,inner sep=2pt]{g}
		&\phantom{}
		& B
       	\end{tikzcd}
    	\]
	Making the two triangles on the left into a square and applying 
	the openness of $f$, we obtain:
   \[
  \begin{tikzcd}
            Jv
            \arrow{rr}[description,inner sep=2pt]{c}
            \arrow[dashed]{dr}[description,inner sep=2pt]{\JE e'}
            \arrow{dd}[description,inner sep=2pt]{\JE e}
            \descto[pos=0.45]{drr}{\commutes}
            \descto[pos=0.75]{ddrr}{\commutes}
            & & C
            \arrow{dd}[description,inner sep=2pt]{f}
            \\
            \phantom{}
            \descto[pos=0.5]{r}{\commutes}
            &Ju
            \arrow[dashed]{ur}[description,inner sep=2pt]{c'}
            \arrow[dashed]{dl}[description,inner sep=2pt]{\JS s}
            & \phantom{}
            \\
            Jw
            \arrow{rr}[description,inner sep=2pt]{a'}
            & & A
        \end{tikzcd}
  \]
	We then obtain:
	\[
	\begin{tikzcd}
		Jw
		\arrow{dd}[description,inner sep=2pt]{a'}
		& \phantom{}
		& \phantom{}
		\descto[pos=0.5]{d}{\commutes}
		& \phantom{}
		& 
		Ju
		\arrow{llll}[description,inner sep=2pt]{\JS s}
		\arrow{dd}[description,inner sep=2pt]{g\cdot c'}
		\\
		\phantom{}
		\descto[pos=0.4]{rr}{\commutes}
		& \phantom{}
		& Jv
		\arrow{llu}[description,inner sep=2pt]{\JE e}
		\arrow{rru}[description,inner sep=2pt]{\JE e'}
		\arrow{lld}[description,inner sep=2pt]{a}
		\arrow{rrd}[description,inner sep=2pt]{b}
		\arrow{d}[description,inner sep=2pt]{c}
		\descto[pos=0.6]{dl}{\commutes}
		\descto[pos=0.6]{dr}{\commutes}
		\descto[pos=0.6]{rr}{\commutes}
		& \phantom{}
		& \phantom{}
		\\
		A
		& \phantom{}
		& C
		\arrow{ll}[description,inner sep=2pt]{f}
		\arrow{rr}[description,inner sep=2pt]{g}
		&\phantom{}
		& B
       	\end{tikzcd}
    	\]
	\item \textbf{backward closure:} obvious.
\end{itemize}
 \qed
\end{proofappx}

\begin{proofappx}{prop:generators}
It is enough to prove the following.
\begin{enumerate}
	\item Assume that we have the following situation:
   \[
  \begin{tikzcd}
            Jv
            \arrow{rr}[description,inner sep=2pt]{b}
            \arrow{d}[description,inner sep=2pt]{\JE e}
            \descto[pos=0.5]{ddrr}{\commutes}
            & & B
            \arrow{dd}[description,inner sep=2pt]{f}
            \\
            Jw
            \arrow{d}[description,inner sep=2pt]{\JE e'}
            & \phantom{}
            & \phantom{}
            \\
            Jw'
            \arrow{rr}[description,inner sep=2pt]{a}
            & & A
        \end{tikzcd}
  \]
 where $e, e' \in \E'$. Applying the assumption on the following situation (in plain):
    \[
  \begin{tikzcd}
            Jv
            \arrow{rr}[description,inner sep=2pt]{b}
            \arrow[dashed]{dr}[description,inner sep=2pt]{\JE e''}
            \arrow{dd}[description,inner sep=2pt]{\JE e}
            \descto[pos=0.45]{drr}{\commutes}
            \descto[pos=0.75]{ddrr}{\commutes}
            & & B
            \arrow{dd}[description,inner sep=2pt]{f}
            \\
            \phantom{}
            \descto[pos=0.5]{r}{\commutes}
            &Ju
            \arrow[dashed]{ur}[description,inner sep=2pt]{b'}
            \arrow[dashed]{dl}[description,inner sep=2pt]{\JS s}
            & \phantom{}
            \\
            Jw
            \arrow{rr}[description,inner sep=2pt]{a\cdot\JE e'}
            & & A
        \end{tikzcd}
  \]
  we obtain the dashed data with $e'' \in \E'$. 
  By assumption, there are $\widetilde{e}' \in \E'$ and $\widetilde{s} \in \bbS$ 
  such that
  \[
  	\JE e'\cdot\JS s = \JS \widetilde{s} \cdot \JE \widetilde{e}'.
  \]
  Rearranging the bottom-right triangle, we obtain the following situation (in plain)
    \[
  \begin{tikzcd}
            Ju
            \arrow{rr}[description,inner sep=2pt]{b'}
            \arrow[dashed]{dr}[description,inner sep=2pt]{\JE e'''}
            \arrow{dd}[description,inner sep=2pt]{\JE\widetilde{e}'}
            \descto[pos=0.45]{drr}{\commutes}
            \descto[pos=0.75]{ddrr}{\commutes}
            & & B
            \arrow{dd}[description,inner sep=2pt]{f}
            \\
            \phantom{}
            \descto[pos=0.5]{r}{\commutes}
            &J\widetilde{w}
            \arrow[dashed]{ur}[description,inner sep=2pt]{b''}
            \arrow[dashed]{dl}[description,inner sep=2pt]{\JS s'}
            & \phantom{}
            \\
            Jw
            \arrow{rr}[description,inner sep=2pt]{a\cdot\JS\widetilde{s}}
            & & A
        \end{tikzcd}
  \]
  and applying the assumption, we obtain the dashed data with $e''' \in \E'$.
  In total, we have obtained:
    \[
  \begin{tikzcd}
            Jv
            \arrow{rr}[description,inner sep=2pt]{b}
            \arrow[dashed]{dr}[description,inner sep=2pt]{\JE(e'''\cdot e'')}
            \arrow{dd}[description,inner sep=2pt]{\JE(e'\cdot e)}
            \descto[pos=0.45]{drr}{\commutes}
            \descto[pos=0.75]{ddrr}{\commutes}
            & & B
            \arrow{dd}[description,inner sep=2pt]{f}
            \\
            \phantom{}
            \descto[pos=0.5]{r}{\commutes}
            &J\widetilde{w}
            \arrow[dashed]{ur}[description,inner sep=2pt]{b''}
            \arrow[dashed]{dl}[description,inner sep=2pt]{\JS(\widetilde{s}\cdot s')}
            & \phantom{}
            \\
            Jw
            \arrow{rr}[description,inner sep=2pt]{a}
            & & A
        \end{tikzcd}
  \]
  with $e'''\cdot e'' \in \E$ and $\widetilde{s}\cdot s' \in \bbS$.
	\item Assume that we have a set $R$ of spans, and that we have the following 
	situation:
	
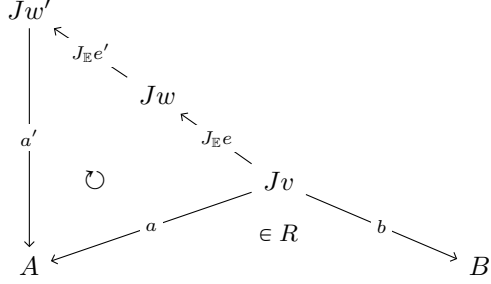
\[
	\begin{tikzcd}
		Jw'
		\arrow{ddd}[description,inner sep=2pt]{a'}
		& \phantom{}
		& \phantom{}
		& \phantom{}
		& \phantom{}\\
		\phantom{}
		& Jw
		\arrow{lu}[description,inner sep=2pt]{\JE e'}
		\descto[pos=0.5]{ddl}{\commutes}
		& \phantom{}
		& \phantom{}
		& \phantom{}\\
		\phantom{}
		& \phantom{}
		& Jv
		\arrow{lu}[description,inner sep=2pt]{\JE e}
		\arrow{lld}[description,inner sep=2pt]{a}
		\arrow{rrd}[description,inner sep=2pt]{b}
		\descto[pos=0.5]{d}{\ensuremath{\in R}\xspace}
		& \phantom{}
		& \phantom{}\\
		A
		& \phantom{}
		&\phantom{}
		& \phantom{}
		& B
       	\end{tikzcd}
    	\]
	with $e, e' \in \E'$.
	Applying the assumption of the following situation (in plain):
	
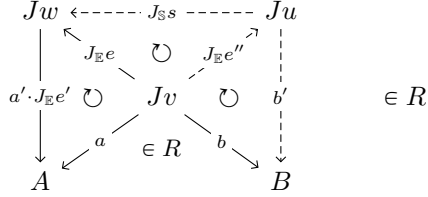
\[
	\begin{tikzcd}
		Jw
		\arrow{dd}[description,inner sep=2pt]{a'\cdot\JE e'}
		& \phantom{}
		\descto[pos=0.5]{d}{\commutes}
		& Ju
		\arrow[dashed]{ll}[description,inner sep=2pt]{\JS s}
		\arrow[dashed]{dd}[description,inner sep=2pt]{b'}
		&\phantom{}\\
		\phantom{}
		\descto[pos=0.5]{r}{\commutes}
		& Jv
		\arrow{lu}[description,inner sep=2pt]{\JE e}
		\arrow[dashed]{ru}[description,inner sep=2pt]{\JE e''}
		\arrow{ld}[description,inner sep=2pt]{a}
		\arrow{rd}[description,inner sep=2pt]{b}
		\descto[pos=0.5]{r}{\commutes}
		\descto[pos=0.5]{d}{\ensuremath{\in R}\xspace}
		& \phantom{}
		& \in R\\
		A
		&\phantom{}
		& B
		&\phantom{}
       	\end{tikzcd}
    	\]
	we obtain the dashed data with $e' \in \E'$ and 
	$A \xleftarrow{~a'' =\, a'\cdot\JE e'\cdot \JS s~} Ju \xrightarrow{~b'~} B \in R$.
	By assumption, there are $\widetilde{e}' \in \E'$ and $\widetilde{s} \in \bbS$ 
  such that
  \[
  	\JE e'\cdot\JS s = \JS \widetilde{s} \cdot \JE \widetilde{e}'.
  \]
  Rearranging this data, we have the following situation (in plain):
	
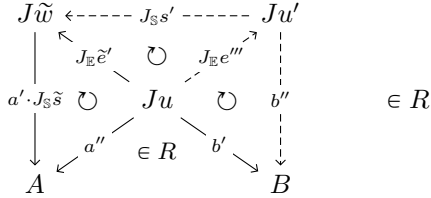
\[
	\begin{tikzcd}
		J\widetilde{w}
		\arrow{dd}[description,inner sep=2pt]{a'\cdot\JS\widetilde{s}}
		& \phantom{}
		\descto[pos=0.5]{d}{\commutes}
		& Ju'
		\arrow[dashed]{ll}[description,inner sep=2pt]{\JS s'}
		\arrow[dashed]{dd}[description,inner sep=2pt]{b''}
		&\phantom{}\\
		\phantom{}
		\descto[pos=0.5]{r}{\commutes}
		& Ju
		\arrow{lu}[description,inner sep=2pt]{\JE\widetilde{e}'}
		\arrow[dashed]{ru}[description,inner sep=2pt]{\JE e'''}
		\arrow{ld}[description,inner sep=2pt]{a''}
		\arrow{rd}[description,inner sep=2pt]{b'}
		\descto[pos=0.5]{r}{\commutes}
		\descto[pos=0.5]{d}{\ensuremath{\in R}\xspace}
		& \phantom{}
		& \in R\\
		A
		&\phantom{}
		& B
		&\phantom{}
       	\end{tikzcd}
    	\]
	with $e''' \in \E'$ and 
	$A \xleftarrow{~a'\cdot\JS\widetilde{s}\cdot \JS s'~} Ju' \xrightarrow{~b''~} B \in R$.
	In total, we have the following situation:
	
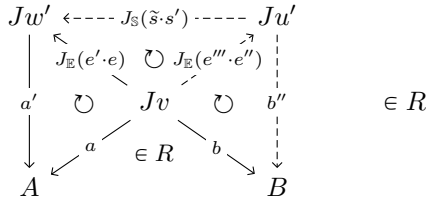
\[
	\begin{tikzcd}
		Jw'
		\arrow{dd}[description,inner sep=2pt]{a'}
		& \phantom{}
		\descto[pos=0.5]{d}{\commutes}
		& Ju'
		\arrow[dashed]{ll}[description,inner sep=2pt]{\JS(\widetilde{s}\cdot s')}
		\arrow[dashed]{dd}[description,inner sep=2pt]{b''}
		&\phantom{}\\
		\phantom{}
		\descto[pos=0.5]{r}{\commutes}
		& Jv
		\arrow{lu}[description,inner sep=2pt]{\JE(e'\cdot e)}
		\arrow[dashed]{ru}[description,inner sep=2pt]{\JE(e'''\cdot e'')}
		\arrow{ld}[description,inner sep=2pt]{a}
		\arrow{rd}[description,inner sep=2pt]{b}
		\descto[pos=0.5]{r}{\commutes}
		\descto[pos=0.5]{d}{\ensuremath{\in R}\xspace}
		& \phantom{}
		& \in R\\
		A
		&\phantom{}
		& B
		&\phantom{}
       	\end{tikzcd}
    	\]
	with $e'''\cdot e'' \in \E$, $\widetilde{s}\cdot s' \in \bbS$, and 
	$A \xleftarrow{~a'\cdot\JS(\widetilde{s}\cdot s')~} Ju' \xrightarrow{~b''~} B \in R$.
	\item The analogue statement for the backward closure is obvious.
\end{enumerate}
\qed
\end{proofappx}

\section{Proofs of Section~\ref{sec:open-weighted}}
 
 \begin{proofappx}{lem:orderMonoidFunctor}
Assume given $f_1, f_2\colon X\to \weighted{(K,+,e)}{A\times Y}$, 
$g\colon X'\to X$, and $h\colon Y\to Y'$ with $f_1 \sqsubseteq f_2$.
Let us prove that 
$\weighted{(K,+,e)}{A\times h}\cdot f_1\cdot g \sqsubseteq 
	\weighted{(K,+,e)}{A\times h}\cdot f_2\cdot g$, 
that is, for all $x'$, $y'$, $a$
\[
	\sum\limits_{\{y \mid h(y) = y'\}} f_1(g(x'))(a,y) \sqsubseteq 
		\sum\limits_{\{y \mid h(y) = y'\}} f_2(g(x'))(a,y).
\]
If $\{y \mid h(y) = y'\}$ is empty, this means proving $e \sqsubseteq e$, which is true by reflexivity.
Otherwise, for all $y$ in this set, 
\[
	f_1(g(x'))(a,y) \sqsubseteq f_2(g(x'))(a,y)
\]
by assumption.
Then, we conclude by monotonicity of $+$.
\qed
\end{proofappx}

\begin{proofappx}{lem:rearrangementMonoids}
For $x = (x_1, \ldots, x_d)$ and $y = (y_1,\ldots,y_d)$ in $\mathbb{R}^d$, we will
denote \[(\min(x_1,y_1),\ldots,\min(x_d,y_d))\] by $x \sqcap y$.
Let us prove it for $(G_{\geq 0},+,0,\leq)$ first.
Define 
\[
	K_x = \{(i,k) \mid x_i = (x_{i,1}, \ldots, x_{i,p}) \wedge x_{i,k} > 0\}
\] 
and $k_x$ the cardinal of $K_x$. We define $K_y$ and $k_y$ similarly.
Let us prove it by induction on $k_x+k_y$.
\begin{itemize}
	\item	if $k_x = 0$, then $u_{i,j} = 0$ works.
	\item if $k_y = 0$, then $\sum\limits_{i=1}^n x_i = 0$, and since $x_i \geq 0$ for all $i$, 
		then $x_i = 0$ for all $i$. This means that $k_x = 0$ and we are back to the previous case.
	\item if $k_x > 0$ and $k_y > 0$, then let $i_0$ and $k$ such that $x_{i_0,k} > 0$. Since 
		$\sum\limits_{i=1}^n x_i  \leq
		\sum\limits_{j=1}^m y_j$, 
		then there exists $j_0$ such that $y_{j_0,k} > 0$.
		By assumption, $z_0 = x_{i_0,k} \sqcap y_{j_0,k}$, 
		$x_{i_0}' =  x_{i_0} - z_0$, and $y_{j_0}' =  y_{j_0} - z_0$ 
		belong to $G_{\geq 0}$. 
		For $i \neq i_0$, define $x_i' = x_i$, 
		and for $j \neq j_0$, define $y_j' = y_j$. Then we have 
		$\sum\limits_{i=1}^n x_i'  \leq
			\sum\limits_{j=1}^m y_j'$.
		By construction, if $x_{i_0,l} = 0$ then $x_{i_0,l}' = 0$, 
		if $y_{j_0,l} = 0$ then $y_{j_0,l}' = 0$, and  
		($x_{i_0,k}' = 0$ or $y_{j_0,k}' = 0$), so then $k_x'+k_y' < k_x+k_y$.
		By induction hypothesis, there is a family $(u_{i,j}')$ in $G_{\geq 0}$ such that
		\[
			\text{for all }j,\,\sum\limits_{i=1}^n u_{i,j}' \leq y_j' ~\text{and}~ 
			\text{for all }i,\,\sum\limits_{j=1}^m u_{i,j}' = x_i'.
		\]
		Define $u_{i,j} = u_{i,j}'$ if $i \neq i_0$ and $j \neq j_0$, $u_{i_0,j_0} + z_0$ otherwise.
\end{itemize}

For $(G,+,0,\leq)$, we reduce the problem $\sum\limits_{i=1}^n x_i  \leq \sum\limits_{j=1}^m y_j$
in $G$ to a problem in $G_{\geq 0}$ as follows.
First, for $x \in G$, define $x^- = -(0 \sqcap x)$ and $x^+ = x + x^-$.
By assumption, $x^-$ and $x^+$ belong to $G_{\geq 0}$ and $x = x^+ - x^-$.
Then, defining $x_i' = x_i^+$ if $i \leq n$, $y_{i-n}^-$ otherwise, and 
$y_j' = y_j^+$ if $j \leq m$, $x_{j-m}^-$ otherwise, we have the following problem in $G_{\geq 0}$
\[
	\sum\limits_{i=1}^{n+m} x_i'  \leq \sum\limits_{j=1}^{m+n} y_j'.
\]
By the previous case, we have $(u_{i,j}')$ in $G_{\geq 0}$ such that 
\begin{itemize}
	\item for $i \leq n$, $\sum\limits_{j=1}^{m+n} u_{i,j} = x_i^+$,
	\item for $i > n$, $\sum\limits_{j=1}^{m+n} u_{i,j} = y_{i-n}^-$,
	\item for $j \leq m$, $\sum\limits_{i=1}^{n+m} u_{i,j} \leq y_j^+$, and
	\item for $j > m$, $\sum\limits_{i=1}^{n+m} u_{i,j} \leq x_{j-m}^-$.
\end{itemize}
Define 
\[
\begin{tabular}{m{3em}|rlr}
	$u_{i,j}$ = 
		& $u'_{i,j} $
		& $- u'_{n+j,m+i}$ 
		& \\
		&
		&
		& if $i > 1$ and $j > 1$
		\\
		& $u'_{1,j}$
		& $- u'_{n+j,m+1}$ 
		& \\
		&
		& $+ \sum\limits_{j'=1}^{m} \left(u'_{n+j',j} - u'_{n+j,j'}\right)$
		& if $i=1$ and $j>1$
		\\
		& $u'_{i,1}$ 
		& $- u'_{n+1,m+i} + \sum\limits_{i'=1}^{n} u'_{i,m+i'}$ 
		& \\
		&
		& $-\left(x_i^- - \sum\limits_{j'=1}^{m} u'_{n+j',m+i}\right)$
		&if $i > 1$ and $j = 1$
		\\
		& $u'_{1,1}$ 
		& $- u'_{n+1,m+1} 
			+ \sum\limits_{j'=1}^{m} \left(u'_{n+j',1} - u'_{n+1,j'}\right)$
		& \\
		& 
		& $+ \sum\limits_{i'=1}^{n} u'_{1,m+i'} -
			\left(x_1^- - \sum\limits_{j'=1}^{m} u'_{n+j',m+1}\right)$ 
		&\quad\quad if $i = 1$ and $j = 1$
\end{tabular}
\]
Let us check that this is a solution of the original problem. First, all the $u_{i,j}$ 
belong to $G$ by assumption. Next, for $i > 1$
\[
	\sum\limits_{j=1}^{m} u_{i,j} = \sum\limits_{j=1}^{m+n} u'_{i,j} - x_i^- = x_i^+-x_i^- = x_i.
\]	
For $j > 1$
\[
	\sum\limits_{i=1}^{n} u_{i,j} = \sum\limits_{i=1}^{n+m} u'_{i,j} - \sum\limits_{j'=1}^{m+n} u'_{n+j,j'} 
		\leq y_j^+ - y_j^- = y_j.
\]
For $i = 1$
\[
	\sum\limits_{j=1}^{m} u_{1,j} = x_1 + \sum\limits_{j=1}^{m}\sum\limits_{j'=1}^{m} u'_{n+j,j'}
		- \sum\limits_{j=1}^{m}\sum\limits_{j'=1}^{m} u'_{n+j',j} = x_1.
\]
For $j = 1$
\[
	\sum\limits_{i=1}^{n} u_{i,1} \leq y_1 + \sum\limits_{i=1}^{n}\sum\limits_{i'=1}^{n} u'_{i,m+i'}
		- \sum\limits_{i=1}^{n}\left(x_i^- - \sum\limits_{j'=1}^{m} u'_{n+j',m+i}\right).
\]
We know that for all $i$, $x_i^- - \sum\limits_{j'=1}^{m} u'_{n+j',m+i} \geq \sum\limits_{i'=1}^{n} u'_{i',m+i}$, which means
\[
	\sum\limits_{i=1}^{n} u_{i,1} \leq y_1 + \sum\limits_{i=1}^{n}\sum\limits_{i'=1}^{n} u'_{i,m+i'}
		- \sum\limits_{i=1}^{n}\sum\limits_{i'=1}^{n} u'_{i',m+i} = y_1.
\]
For strictness, let us prove something stronger, namely that if $K$ is a 
rearrangement monoid and $+$ is strictly monotone (which is the case in $\mathbb{R}^d$), 
then it is strict.
Indeed, assume $\sum\limits_{i=1}^n x_i  = \sum\limits_{j=1}^m y_j$, 
then we have by rearrangement 
that $\sum\limits_{i=1}^n u_{i,j}  \leq y_j$ and $\sum\limits_{j=1}^m u_{i,j} = x_i$, so
\[
	\sum\limits_{j=1}^m\sum\limits_{i=1}^n u_{i,j} = \sum\limits_{i=1}^n x_i 
		= \sum\limits_{j=1}^m y_j.
\]
Since the sum is strictly monotone, if $\sum\limits_{i=1}^n u_{i,j}  < y_j$ for some 
$j$ then, $\sum\limits_{j=1}^m\sum\limits_{i=1}^n u_{i,j} < \sum\limits_{j=1}^m y_j$ 
which is a contradiction. Consequently, $\sum\limits_{i=1}^n u_{i,j}  = y_j$ for all j.

Assume given a distributive lattice with bottom element.
Define $u_{i,j} = x_i \sqcap y_j$.
For all $j$, $u_{i,j} \leq y_j$, and thus $\bigsqcup\limits_{i=1}^n u_{i,j} \leq y_j$, 
which is the first condition.
By the assumption, for all $i$
\[
	x_i \leq \bigsqcup\limits_{j=1}^m y_j.
\]
But since $L$ is distributive
\[
	x_i = x_i \sqcap \left(\bigsqcup\limits_{j=1}^m y_j\right) = 
		\bigsqcup\limits_{j=1}^m \left(x_i\sqcap y_j\right) = \bigsqcup\limits_{j=1}^m u_{i,j},
\]
which is the second condition.

If furthermore $\sum\limits_{i=1}^n x_i  = \sum\limits_{j=1}^m y_j$,
then we can prove symmetrically, using distributivity, that $y_j = \bigsqcup\limits_{i=1}^n u_{i,j}$
which proves strictness.
Conversely, assume that $L$ is a rearrangement monoid.
We then want to prove that for all $x$, $y_1$, $y_2$, 
\[
	x \sqcap (y_1 \sqcup y_2) = (x \sqcap y_1) \sqcup (x \sqcap y_2).
\]
Let us start proving it when $x \leq y_1 \sqcup y_2$.
In that case, we can use the rearrangement condition and obtain 
$u_1$ and $u_2$ such that $x = u_1 \sqcup u_2$ and 
$u_j \sqcup y_j$. This implies that $u_j \leq x \sqcap y_j$, so then:
\[
	x \sqcap (y_1 \sqcup y_2) = x = u_1 \sqcup u_2 
		\leq (x \sqcap y_1) \sqcup (x \sqcap y_2) \leq x,
\]
which implies what we want.
Now, let us assume any $x$, $y_1$, $y_2$.
Apply the previous case with $x \sqcap (y_1 \sqcup y_2)$, $y_1$, $y_2$.
We then obtain 
\begin{align*}
	x \sqcap (y_1 \sqcup y_2) 
		& = (x \sqcap (y_1 \sqcup y_2)) \sqcap (y_1 \sqcup y_2)\\
		& = (x \sqcap (y_1 \sqcup y_2) \sqcap y_1) \sqcup (x \sqcap (y_1 \sqcup y_2) \sqcap y_2)\\
		& = (x \sqcap y_1) \sqcup (x \sqcap y_2).
\end{align*}
\qed
\end{proofappx}

\begin{proofappx}{thm:weightedOpen}
Assume that $f\colon (X,c,i)\to (Y,d,j)$ is $(\E_K,\bbS_K)$-open, $X$ is reachable, 
and $K$ is positive.
Proving that $f$ is a proper homomorphism means proving
  \[
    \begin{tikzcd}
      1
      \arrow{r}[description,inner sep=2pt]{i}
      \arrow[bend right=40]{rd}[description,inner sep=2pt]{j}
      \descto{dr}{$\commutes$}
      & X
      \arrow{r}[description,inner sep=2pt]{c}
      \arrow{d}[description,inner sep=2pt]{f}
      \descto{dr}{\commutes}
      & \weighted{K}{A\times X}
      \arrow{d}[description,inner sep=2pt]{\weighted{K}{A\times f}}
      \\
      & Y
      \arrow{r}[description,inner sep=2pt]{d}
      & \weighted{K}{A\times Y}
    \end{tikzcd}
  \]
  The left triangle is given by the fact that $f$ is a lax homomorphism.
  The right square means that for all $x \in X$, for all $y \in Y$, and for $a \in A$,
  \[
  	d(f(x))(a,y) = \sum\limits_{\{x' \mid f(x') = y\}} c(x)(a,x').
  \]
  Since $f$ is a lax homomorphism, we already know that 
  \[
  	d(f(x))(a,y) \sqsupseteq \sum\limits_{\{x' \mid f(x') = y\}} c(x)(a,x').
  \]
  and since $\sqsubseteq$ is antisymmetric, it is enough to prove the other inequality.
  If $d(f(x))(a,y) = 0$, then by positivity, we are done.
  Let us assume given $x$, $y$, and $a$ such that $d(f(x))(a,y) > 0$.
  Since $x$ is reachable, there is a morphism $\alpha\colon Jw\to X$ where $w$ is a word
  $(a_1,k_1),\ldots,(a_n,k_n)$
  on $A\times (K\setminus\{e\})$ and $\text{end}(\alpha) = x$.
  Consider $v = (w,a,d(f(x))(a,y))$ so that there is an edge $e\colon w\to v$ in $\E$.
  There is a lax homomorphism $\beta\colon Jv\to Y$ defined as $\beta(i) = f\cdot\alpha(i)$ for 
  $i \leq n$ and $\beta(n+1,1) = y$.
  So we have the following situation (in plain)
   \[
  \begin{tikzcd}
            Jw
            \arrow{rr}[description,inner sep=2pt]{\alpha}
            \arrow[dashed]{dr}[description,inner sep=2pt]{\JE e}
            \arrow{dd}[description,inner sep=2pt]{\JE e}
            \descto[pos=0.45]{drr}{\commutes}
            \descto[pos=0.75]{ddrr}{\commutes}
            & & (X,c,i)
            \arrow{dd}[description,inner sep=2pt]{f}
            \\
            \phantom{}
            \descto[pos=0.5]{r}{\commutes}
            &Ju
            \arrow[dashed]{ur}[description,inner sep=2pt]{\alpha'}
            \arrow[dashed]{dl}[description,inner sep=2pt]{\JS s}
            & \phantom{}
            \\
            Jv
            \arrow{rr}[description,inner sep=2pt]{\beta}
            & & (Y,d,j)
        \end{tikzcd}
  \]
  Since $f$ is $(\E_K,\bbS_K)$-open, then we obtain that dashed data.
  In particular, $u = (w,a,w_2)$ with $w_2 = l_1,\ldots,l_m$ with $m \geq 1$ and 
  \[
  	d(f(x))(a,y) = \sum\limits_{i = 1}^{m} l_i.
  \]
  Furthermore, for all $i$, $\beta\cdot \JS s(n+1,i) = y$.
  Since $\alpha'$ is a lax morphism, for all $i$
  \[
  	c(x)(a,i) \sqsupseteq \sum\limits_{\{i' \mid \alpha'(i') = i\}} l_{i'}.
  \]
  Consequently, since for all $i$, $f\cdot\alpha'(n+1,i) = \beta\cdot \JS s(n+1,i) = y$
  \[
  	d(f(x))(a,y) 
		= \sum\limits_{i = 1}^{m} l_i 
		\sqsubseteq \sum\limits_{\{x' \mid \exists i.\, \alpha'(n+1,i) = x'\}} c(x)(a,x')
		\sqsubseteq \sum\limits_{\{x' \mid f(x') = y\}} c(x)(a,x'),
  \]
  the last inequality being by positivity.
  
  Conversely, assume that $f$ is a proper homomorphism and that $K$ is a rearrangement monoid.
  Let us prove that $f$ is $(\E_K,\bbS_K)$-open. So we start with the following situation (in plain)
   \[
  \begin{tikzcd}
            Jw
            \arrow{rr}[description,inner sep=2pt]{\alpha}
            \arrow[dashed]{dr}[description,inner sep=2pt]{\JE e}
            \arrow{dd}[description,inner sep=2pt]{\JE e}
            \descto[pos=0.45]{drr}{\commutes}
            \descto[pos=0.75]{ddrr}{\commutes}
            & & (X,c,i)
            \arrow{dd}[description,inner sep=2pt]{f}
            \\
            \phantom{}
            \descto[pos=0.5]{r}{\commutes}
            &Ju
            \arrow[dashed]{ur}[description,inner sep=2pt]{\alpha'}
            \arrow[dashed]{dl}[description,inner sep=2pt]{\JS s}
            & \phantom{}
            \\
            Jv
            \arrow{rr}[description,inner sep=2pt]{\beta}
            & & (Y,d,j)
        \end{tikzcd}
  \]
  with $w=(a_1,k_1),\ldots,(a_n,k_n)$ being a word on $A\times K\setminus\{e\}$ 
  and $v$ being of the form 
  $(w,a,w_2)$ with $w_2 = l_1,\dots,l_m$ a word on $K\setminus\{e\}$.
  Let us call $x = \text{end}(\alpha)$ and $L = \{\beta(n+1,i) \mid 1 \leq i \leq m\} \subseteq Y$.
  For $y$ in $L$, define $V_y = \{i \mid \beta(n+1,i) = y$.
  Since $\beta$ is a lax homomorphism, 
  \[
  	d(f(x))(a,y) \sqsupseteq \sum\limits_{i \in V_y} l_i.
  \]
  Denote the set $\{x' \in X \mid f(x') = y \wedge c(x)(a,x') \neq e\}$ by $X_y$. 
  Since $f$ is a proper homomorphism,
  \[
  	d(f(x))(a,y) = \sum\limits_{x' \in X_y} c(x)(a,x').
  \]
  Since $V_y$ and $X_y$ are finite subsets of $K$, then since $K$ is a rearrangement 
  monoid, there is a family $(u_{i,x'})_{i \in V_y,x' \in X_y}$ such that
  \[
  	\sum\limits_{x' \in X_y} u_{i,x'} = l_i \text{ and } 
		\sum\limits_{i \in V_y} u_{i,x'} \sqsubseteq c(x)(a,x').
  \]
  Define $U_i = \{x' \in X_y \mid u_{i,x'} \neq e\} = \{x_1,\ldots,x_{m_i}\}$ for $i \in V_y$, 
  and $m_i$ its cardinal. 
  Since $l_i \neq e$, $U_i$ is non-empty.
  Now, define $u = (w,a,w_2')$ with $w_2' = l_{1,1},\ldots,l_{1,m},\ldots,l_{m,1},\ldots,l_{m,m_m}$
  with $l_{i,j} = u_{i,x_j} \neq e$.
  There is an edge $e\colon w\to u$ in $\E$.
  Define $s\colon\{1,\ldots,\sum\limits_{i = 1}^{m} m_i\}\to\{1,\ldots,m\}$ as
  $s(k) = j$ if $\sum\limits_{i = 1}^{j-1} m_i < k \leq \sum\limits_{i = 1}^{j} m_i$.
  This function is monotone (easy) and surjective (since $U_i$ is non-empty).
  Now for any $j$
  \[
  	\sum\limits_{\left\{k \mid s\left(k+\sum_{i = 1}^{j-1} m_i\right) = j\right\}} l_{j,k}
		= \sum\limits_{k=1}^{m_j} l_{j,k}
		= \sum\limits_{k=1}^{m_j} u_{i,x_k}
		= l_j,
  \]
  so that $s$ is a morphism from $v$ to $u$ in $\bbS$.
  It is clear that $\JS s\cdot \JE e = \JE e$.
  
  Define $\alpha'\colon Ju\to X$ as $\alpha'(i) = \alpha(i)$ for $i \leq n$, 
  and $\alpha'(n+1,k) = x_s$ for $\sum\limits_{i = 1}^{j-1} m_i < k \leq \sum\limits_{i = 1}^{j} m_i$
  and $s = k - \sum\limits_{i = 1}^{j-1} m_i$.
  By construction, $\alpha'$ is a lax morphism and $\alpha'\cdot \JE e$.
  Then, for $k$ with $s(k) \in V_y$, 	
  $f\cdot\alpha'(n+1,k) = f(x')$ with $x' \in X_y$ and 
  \[
  	\beta\cdot \JS s(n+1,k) = y = f\cdot\alpha'(n+1,k).
  \]
  Furthermore, since $f\cdot\alpha = \beta\cdot \JE e$, we obtain that 
  \[
  	\beta\cdot \JS s = f\cdot \alpha',
  \]
  which concludes that $f$ is $(\E_K,\bbS_K)$-open.
  
  For the last statement, by the previous two points, the only thing to prove is 
  that when there is a span 
  $(X,c,i) \xleftarrow{~f~} (Z,e,k) \xrightarrow{~g~} (Y,d,j)$ 
  of $(\E_K,\bbS_K)$-open maps, then there is one where 
  $(Z,e,k)$ is reachable.
  
  Define $(\widetilde{Z},\widetilde{e},k)$ with $\widetilde{Z}$ being the set of states 
  of $Z$ that are reachable. In particular, $k$ is reachable.
  Define $\widetilde{e}(z_1)(a,z_2) = e(z_1)(a,z_2)$ for $z_1$, $z_2$ reachable.
  It is enough to prove that the injection $\iota\colon\widetilde{Z}\to Z$ is 
  a proper morphism. Indeed, in that case, by what precedes, $\iota$ is 
  $(\E_K,\bbS_K)$-open, and since open maps are closed under composition by 
  Lemma~\ref{lem:open-composition}, then 
  \[
  	(X,c,i) \xleftarrow{~f\cdot\iota~} (\widetilde{Z},\widetilde{e},k) 
		\xrightarrow{~g\cdot\iota~} (Y,d,j)
  \]
  is a span of open maps whose tip is reachable.
  Now, $\iota$ is a proper morphism for the following reason.
  For $z$ reachable and $z'$ non-reachable,
  $e(z_1)(a,z_2) = 0$ for any $a$, otherwise $z'$ would be reachable.
  Consequently, for $z$ reachable, $z'$ non-reachable
  \begin{align*}
  	e(\iota(z))(a,z') = 0 & = \sum\limits_{a, z'' \in \varnothing} \widetilde{e}(z)(a,z'')\\
		&=  \sum\limits_{a \in A, z'' \text{ reach.}, \iota(z'') = z} \widetilde{e}(z)(a,z'')\\
		&= \weighted{K}{A\times\iota}\cdot \widetilde{e}(z)(a,z')
  \end{align*}
  and for $z'$ reachable
  \[
  	e(\iota(z))(a,z') = e(\iota(z))(a,\iota(z')) = \widetilde{e}(z)(a,z') = \weighted{K}{A\times\iota}\cdot \widetilde{e}(z)(a,z').
  \]
  \qed
\end{proofappx}

\begin{proofappx}{cor:weightedTransitivity}
It is enough to prove that strictness of $K$ implies that the weighted functor preserves 
weak pullbacks. That is, assume given a pullback in $\Set$:
	\[
	\begin{tikzcd}
            W
            \pullbackangle[dashed]{-45}
            \arrow{r}[description,inner sep=2pt]{k}
            \arrow{d}[description,inner sep=2pt]{h}
            \descto{dr}{\commutes}
            & Y
            \arrow{d}[description,inner sep=2pt]{g}
            \\
            X
            \arrow{r}[description,inner sep=2pt]{f}
            & Z
        \end{tikzcd}
	\]
and form this other pullback:
	\[
	\begin{tikzcd}
            W'
            \pullbackangle[dashed]{-45}
            \arrow{rr}[description,inner sep=2pt]{k'}
            \arrow{d}[description,inner sep=2pt]{h'}
            & \phantom{}
            \descto{d}{\commutes}
            & \weighted{K}{A\times Y}
            \arrow{d}[description,inner sep=2pt]{\weighted{K}{A\times g}}
            \\
            \weighted{K}{A\times X}
            \arrow{rr}[description,inner sep=2pt]{\weighted{K}{A\times f}}
            & \phantom{}
            & \weighted{K}{A\times Z}
        \end{tikzcd}
	\]
Since $\Set$ has the external axiom of choice (that is, that every epi is split), 
it is enough to prove that the unique morphism $e\colon\,\weighted{K}{A\times W}\to W'$ 
such that
\[
	h'\cdot e = h \quad\quad\text{and}\quad\quad k'\cdot e = k
\]
is a surjective function.
Let us describe more concretely $W$ and $W'$:
\[
	W = \left\{(x,y) \in X\times Y \mid f(b) = g(c)\right\}
\]
\[
	W' = \left\{(\mu,\rho) \mid
		\forall (a,z)\in A\times Z.\, 
		\sum\limits_{\{x \in X \mid f(x) = z\}} \mu(a,x) = 
		\sum\limits_{\{y \in Y \mid g(y) = z\}} \rho(a,y)\right\}.
\]
Now the function $e$ is given by 
\[
	\theta  \mapsto
		\left((a,x) \mapsto \sum\limits_{\{y \in Y \mid (x,y) \in W\}} \theta(a,(x,y)),
			(a,y) \mapsto \sum\limits_{\{x \in X \mid (x,y) \in W\}} \theta(a,(x,y))\right)
\]
Fix $(\mu,\rho) \in W'$. We want to construct $\theta \in\weighted{K}{A\times W}$ 
such that $e(\theta) = (\mu,\rho)$.
By definition of the the weighted functor, 
the set 
\[
	\mathcal{N} = \left\{(a,z) \mid \exists x.\, f(x) = z \wedge \mu(a,x) \neq e\right\}
\]
is finite. Then, for fixed $(a,z) \in \mathcal{N}$, the sets
\[
	\mathcal{N}_{a,z,X} = \left\{x \mid f(x) = z \wedge \mu(a,x) \neq e\right\}
\]
and
\[
	\mathcal{N}_{a,z,Y} = \left\{y \mid g(y) = z \wedge \rho(a,y) \neq e\right\}
\]
are finite.
Consequently, the condition that $(\mu,\rho) \in W'$ can be reformulated as
\[
	\sum\limits_{x \in \mathcal{N}_{a,z,X}} \mu(a,x) = 
		\sum\limits_{y \in \mathcal{N}_{a,z,Y}} \rho(a,y).
\]
Now, by strictness of $K$, this means we have $(u_{x,y}^a)_{x \in \mathcal{N}_{a,z,X},y \in \mathcal{N}_{a,z,Y}}$ family of elements of $K$ such that:
\[
	\sum\limits_{x \in \mathcal{N}_{a,z,X}} u_{x,y}^a = \rho(a,y) 
	\quad\quad\text{and}\quad\quad
	\sum\limits_{y \in \mathcal{N}_{a,z,Y}} u_{x,y}^a = \mu(a,x) 
\]
Define $\theta$ as
\[
\begin{tabular}{r|lr}
	$(a,(x,y)) \in W \mapsto$
		& $u_{x,y}^a$ 
		& if $z = f(x)$, $(a,z) \in \mathcal{N}$, $x \in \mathcal{N}_{a,z,X}$,\\
		&
		&  and 
			$y \in \mathcal{N}_{a,z,Y}$
		\\
		& $0$ 
		& otherwise
\end{tabular}
\]
By finiteness of $\mathcal{N}$, $\mathcal{N}_{a,z,X}$, and 
$\mathcal{N}_{a,z,Y}$, $\theta \in \weighted{K}{A\times W}$, and by construction
$e(\theta) = (\mu,\rho)$.
\qed
\end{proofappx}

\begin{proofappx}{lem:from-R-to-D}
Let $y \in Y$, so that $d(y) \in \weighted{(\Real_+,+)}{A\times Y}$.
Let us prove that $d(y) \in \Subdist(A\times Y)$.
Let us consider the case where $f$ is a lax homomorphism (the other is similar).
Then for every $x' \in X$, $\sum\limits_{\{y'\,\mid\, f(y') = x'\}} d(y)(y') \leq c(f(y),x')$.
This means that:
\[
	\sum\limits_{y' \in Y} d(y)(y') = \sum\limits_{x'\in X}\sum\limits_{\{y'\,\mid\, f(y') = x'\}} d(y)(y')
		\leq \sum\limits_{x'\in X}c(f(y),x') \leq 1,
\]
and $d(y) \in \Subdist(A\times Y)$.
\qed
\end{proofappx}

\begin{proofappx}{cor:D-open-iff-R-open}
Let us assume that it is $(\E_{\Dist}, \bbS_{\Dist})$-open, and prove that it is\\ 
$(\E_{(\Real_+,+)}, \bbS_{(\Real_+,+)})$-open. So we have the following situation 
in $\weighted{(\Real_+,+)}{\_}$-coalgebras (in plain):
   \[
  \begin{tikzcd}
            Jv
            \arrow{rr}[description,inner sep=2pt]{y}
            \arrow[dashed]{dr}[description,inner sep=2pt]{\JE e'}
            \arrow{dd}[description,inner sep=2pt]{\JE e}
            \descto[pos=0.45]{drr}{\commutes}
            \descto[pos=0.75]{ddrr}{\commutes}
            & & (Y,d,j)
            \arrow{dd}[description,inner sep=2pt]{f}
            \\
            \phantom{}
            \descto[pos=0.5]{r}{\commutes}
            &Ju
            \arrow[dashed]{ur}[description,inner sep=2pt]{y'}
            \arrow[dashed]{dl}[description,inner sep=2pt]{\JS s}
            & \phantom{}
            \\
            Jw
            \arrow{rr}[description,inner sep=2pt]{x}
            & & (X,c,i)
        \end{tikzcd}
  \]
By Lemma~\ref{lem:from-R-to-D} on $x$ and $y$, this situation is also in 
$\Subdist(A\times \_)$-coalgebras. Then by 
$(\E_{\Dist}, \bbS_{\Dist})$-openness, we open the dashed data
in $\Subdist(A\times \_)$-coalgebras, which are also 
$\weighted{(\Real_+,+)}{\_}$-coalgebras.
Conversely, we start with the same plain situation as above, in  
$\Subdist(A\times \_)$-coalgebras this time, so also in 
$\weighted{(\Real_+,+)}{\_}$-coalgebras.
By $(\E_{(\Real_+,+)}, \bbS_{(\Real_+,+)})$-openness, 
we obtain the dashed data in 
$\weighted{(\Real_+,+)}{\_}$-coalgebras. 
By Lemma~\ref{lem:from-R-to-D} on $y'$, the dashed data is also in 
$\Subdist(A\times \_)$-coalgebras.
  
Finally, if $(X,c,i)$ and $(Y,d,j)$ are $(\E_{\Dist}, \bbS_{\Dist})$-bisimilar, then they are\\ 
$(\E_{(\Real_+,+)}, \bbS_{(\Real_+,+)})$-bisimilar by what precedes.
For the converse, if we have a span $(X,c,i) \xleftarrow{~f~} (Z,e,k) \xrightarrow{~g~} (Y,d,j)$ of 
$(\E_{(\Real_+,+)}, \bbS_{(\Real_+,+)})$-open morphisms, to conclude by what precedes, 
we have to make sure that $(Z,e,k)$ is a $\Subdist(A\times \_)$-coalgebra, which is again by 
Lemma~\ref{lem:from-R-to-D}.
\qed
\end{proofappx}

\section{Proofs of Section~\ref{sec:weak-bisimulations}}
\begin{proofappx}{lem:run-R}
Let us prove first that $\pi_X$ is a morphism.
So we start with a transition
\[ 
	(X \xleftarrow{x} Jv \xrightarrow{y} Y) \xtransto{~a~}_R (X \xleftarrow{x'} Jw \xrightarrow{y'} Y)
\]
of $\widetilde{R}$, and distinguish the two cases.
\begin{itemize}
	\item Assume $a \neq \tau$, $v = (n_1, a_1, \ldots, a_k, n_{k+1})$, 
		$w = (n_0, a_1, \ldots, a_k, n_{k+1}, a, 0)$, 
		$x' = x\cdot \JE e_a$, and we want to prove that there is a 
		transition
		\[
			\text{end}(x) = x(n_{k+1},k+1) \xtransto{~a~} x'(0,k+2) = \text{end}(x').
		\]
		Since $x' = x\cdot \JE e_a$, then $x'(n_{k+1},k+1) = x(n_{k+1},k+1)$, 
		and since $x'$ is a morphism, there is a transition as wanted.
	\item The proof of the other case is the same.
\end{itemize}
  For $r\colon Jv\to \widetilde{R}$, $\runend(r)\in \widetilde R$, i.e.~by definition, $\runend(r)$ is a span of the form
  \[
    (X\xleftarrow{~x~}Jv\xrightarrow{~y~}Y).
  \]
  The statement then follows from a simple induction on $v$.
\end{proofappx}

\begin{proofappx}{lem:restriction-to-strong}
Let us start with the following square
\[
        \begin{tikzcd}
            Jv
            \arrow{r}[description,inner sep=2pt]{x}
            \arrow{d}[description,inner sep=2pt]{\JE e}
            \descto{dr}{\commutes}
            & X
            \arrow{d}[description,inner sep=2pt]{f}
            \\
            Jw
            \arrow{r}[description,inner sep=2pt]{y}
            & Y
        \end{tikzcd}
\]
for which we construct a new square with a strong morphism $x'\colon Jv'\to X$.
The lifting in the new square will then give rise to a lifting in the above
square.

Let us start by constructing $s\colon v\to v' \in \bbS$ and $x'\colon Jv'\to X$\twnote{}
such that $x = x'\cdot \JS s$. Assume that $v$ is of the form
$n_1, a_1, \ldots, n_{k+1}$ and $v'$ will be of the form 
$n_1', a_1, \ldots, n_{k+1}'$ with $n_i' \leq n_i$.
Define $f_i\colon\{0 < \ldots < n_i\}\to\{0 < \ldots < n_i\}$ as
\[
	f_i(j) = \min\set{k \leq j \mid \forall k \leq k' \leq j.\, x(k',i) = x(j,i)}.
\]
Then $n_i'$ is the number of elements in the image of $f_i$, that is,
$n_i'$ is the number of $\tau$-transitions from $(0,i)$ to $(n_i,i)$ 
that are mapped to $\tau$-transitions by $x$.
If we denote by $g_i$ the unique monotone bijection from the 
image of $f_i$ to $\{0 < \ldots < n_i'\}$, then $s_i = g_i\cdot f_i$.
It is clear that $s_i\colon\{0 < \ldots < n_i\}\to\{0 < \ldots < n_i'\}$ 
is monotone and surjective, so $s = (s_1,\dots,s_{k+1})$ is a 
morphism from $v$ to $v'$ in $\bbS$.
Now, define $x'\colon Jv' \to X$ by
\[
	x'(i,j) = x(g_j^{-1}(i),j).
\]
Let us prove that $x'$ is a strong morphism.
In $Jv'$, we have two types of transitions:
\begin{itemize}
	\item $(i,j) \xtransto{~\tau~} (i+1,j)$. By construction of $v'$, it means that we have
		\[
			(g_j^{-1}(i),j) \xtransto{~\tau~} (g_j^{-1}(i)+1,j) \xtransto{~\tau~} \ldots \xtransto{~\tau~} (g_j^{-1}(i)+k,j) \xtransto{~\tau~} (g_j^{-1}(i+1),j)
		\]
	in $Jv$ for some $k$, and that $x(g_j^{-1}(i)+k',j) = x(g_j^{-1}(i),j)$ for all $k' \leq k$, and 
	$x(g_j^{-1}(i)+k,j) \xtransto{~\tau~} x(g_j^{-1}(i+1),j)$ in $X$. This means that 
	$x'(i,j) \xtransto{~\tau~} x'(i+1,j)$ in $X$.
	\item $(n_j',i) \xtransto{~a~} (0,j+1)$. Similarly, we can prove that 
	$x'(n_j',i) \xtransto{~a~} x'(0,j+1)$ in $X$.
\end{itemize}
Now
\[
	x'\cdot \JS(s)(i,j) = x'(s_j(i),j) = x(f_j(i),j) = x(i,j).
\]

The second step is to obtain a diagram of the following shape
\[
        \begin{tikzcd}
        	    \phantom{}
	    & \phantom{}
            \descto[pos=0.85]{d}{\commutes}
            & \phantom{}\\
            Jv
            \arrow[bend left=40]{rr}[description,inner sep=2pt]{x}
            \arrow{d}[description,inner sep=2pt]{\JE e}
            \arrow{r}[description,inner sep=2pt]{\JS s}
            \descto{dr}{\commutes}
            & Jv'
            \arrow{r}[description,inner sep=2pt]{x'}
            \arrow{d}[description,inner sep=2pt]{\JE e'}
            \descto{dr}{\commutes}
            & X
            \arrow{d}[description,inner sep=2pt]{f}
            \\
            Jw
            \arrow[bend right=40]{rr}[description,inner sep=2pt]{y}
            \arrow{r}[description,inner sep=2pt]{\JS s'}
            & Jw'
            \arrow{r}[description,inner sep=2pt]{y'}
            & Y\\
        	    \phantom{}
	    & \phantom{}
            \descto[pos=0.8]{u}{\commutes}
            & \phantom{}
        \end{tikzcd}
\]
For now, let us assume is $e_a$ (the case $e = e_\tau$ is similar).
In that case, $w$ is of the shape $n_1, a_1, \ldots, n_{k+1}'',a,n_{k+2}$
with $n_{k+1} \leq n_{k+1}''$. We can then define $w'$ as
$n_1', a_1, \ldots, n_{k+1}' + n_{k+1}'' - n_{k+1}, a, n_{k+2}$ and
$s' = (s_1', \ldots, s_{k+2}')$ as
\[
\begin{tabular}{m{4em}|lr}
	$s_j'(i)$ = 
		& $s_j(i)$ 
		& if $j \leq k$ or ($j = k+1$ and $i \leq n_{k+1}$)
		\\
		& $n_{k+1}'+i-n_{k+1}$ 
		& if $j = k+1$ and $i > n_{k+1}$
		\\
		& $i$ 
		&if $j = k+2$
\end{tabular}
\]
Then $s'$ is a morphism from $w$ to $w'$ in $\bbS$.
Now, by construction of $v'$ and $w'$, there is a unique 
$e'=e_a\colon v'\to w'$ in $\E$, and it is easy to see that 
$\JE e'\cdot \JS s = \JS s'\cdot \JE e$.
Finally, we want to construct $y'\colon Jw'\to Y$ as
\[
\begin{tabular}{m{4em}|lr}
	$y'(i,j)$ = 
		& $y(g_j^{-1}(i),j)$ 
		& if $j \leq k$ or ($j = k+1$ and $i \leq n_{k+1}'$)
		\\
		& $y(n_{k+1}+i-n_{k+1}',k+1)$ 
		& if $j = k+1$ and $i > n_{k+1}'$
		\\
		& $y(i,k+2)$ 
		&if $j = k+2$
\end{tabular}
\]
Now:
\[
	y'\cdot \JE e'(i,j) = y(g_j^{-1}(i),j) = f\cdot x(g_j^{-1}(i),j) = f\cdot x'(i,j),
\]
which proves the commutativity of the right square. Also, let us prove the bottom commutativity.
The only interesting part is when $j \leq k$ or ($j = k+1$ and $i \leq n_{k+1}$).
In that case, 
\[
	y'\cdot \JS s' = y(f_j(i),j) = f\cdot x(f_j(i),j) = f\cdot x(i,j) = y(i,j).	
\]
Using that $x'$, $f$, and $y$ are morphisms, it is easy to prove that $y'$ is 
also a morphism.

Using the assumption on the right square, since $x'$ is a strong morphism, 
we obtain the following data
   \[
  \begin{tikzcd}
            Jv'
            \arrow{rr}[description,inner sep=2pt]{x'}
            \arrow[dashed]{dr}[description,inner sep=2pt]{\JE e''}
            \arrow{dd}[description,inner sep=2pt]{\JE e'}
            \descto[pos=0.45]{drr}{\commutes}
            \descto[pos=0.75]{ddrr}{\commutes}
            & & X
            \arrow{dd}[description,inner sep=2pt]{f}
            \\
            \phantom{}
            \descto[pos=0.5]{r}{\commutes}
            &Ju'
            \arrow[dashed]{ur}[description,inner sep=2pt]{x''}
            \arrow[dashed]{dl}[description,inner sep=2pt]{\JS s''}
            & \phantom{}
            \\
            Jw'
            \arrow{rr}[description,inner sep=2pt]{y'}
            & & Y
        \end{tikzcd}
  \]
  Since $\JS s''\cdot \JE e'' = \JE e'$ and $e' = e_a$, we know the following:
  \begin{itemize}
  	\item $e'' = e_a$,
	\item $u'$ is of the form $n_1', a_1, \ldots, n_{k+1}''', a, n_{k+2}'$ with 
		$n_{k+1}' \leq n_{k+1}' + n_{k+1}'' - n_{k+1} \leq n_{k+1}'''$ and 
		$n_{k+2} \leq n_{k+2}'$,
	\item for all $j \leq k$ or ($j = k+1$ and $i \leq n_{k+1}'$), $s_j''(i) = i$.
\end{itemize}
To conclude, we want the following
   \[
  \begin{tikzcd}
            Jv
            \arrow{rr}[description,inner sep=2pt]{x}
            \arrow[dashed]{dr}[description,inner sep=2pt]{\JE e'''}
            \arrow{dd}[description,inner sep=2pt]{\JE e}
            \descto[pos=0.45]{drr}{\commutes}
            \descto[pos=0.75]{ddrr}{\commutes}
            & & X
            \arrow{dd}[description,inner sep=2pt]{f}
            \\
            \phantom{}
            \descto[pos=0.5]{r}{\commutes}
            &Ju
            \arrow[dashed]{ur}[description,inner sep=2pt]{x'''}
            \arrow[dashed]{dl}[description,inner sep=2pt]{\JS s'''}
            & \phantom{}
            \\
            Jw
            \arrow{rr}[description,inner sep=2pt]{y}
            & & Y
        \end{tikzcd}
  \]
  We first define $u$ as 
  \[
  	n_1, a_1, \ldots, n_{k+1}'''+n_{k+1}-n_{k+1}', a, n_{k+2}'
  \] 
  Since $n_{k+1}'''+n_{k+1}-n_{k+1}' \geq n_{k+1}$, there is a unique 
  $e''' = e_a\colon v\to u$ in $\E$. Next, $s'''$ is defined as
\[
\begin{tabular}{m{4em}|lr}
	$s'''_j(i)$ = 
		& i 
		& if $j \leq k$ or ($j = k+1$ and $i \leq n_{k+1}$)
		\\
		& $s''_{k+1}(n_{k+1}'+i-n_{k+1}) +$ 
		& 
		\\
		& \qquad$n_{k+1}' - n_{k+1}$ 
		& if $j = k+1$ and $i > n_{k+1}$
		\\
		& $s''_{k+2}(i)$ 
		&if $j = k+2$
\end{tabular}
\]
which is a morphism from $u$ to $w$ in $\bbS$ with $\JS s'''\cdot \JE e''' = \JE e$.
We now need to construct $x'''\colon u\to X$ as 
\[
\begin{tabular}{m{4.7em}|lr}
	$x'''(i,j)$ = 
		& $x(i,j)$
		& if $j \leq k$ or ($j = k+1$ and $i \leq n_{k+1}$)
		\\
		& $x''(n_{k+1}'+i-n_{k+1},k+1)$ 
		& if $j = k+1$ and $i > n_{k+1}$
		\\
		& $x''(i,k+2)$ 
		&if $j = k+2$
\end{tabular}
\]
Since $x''$ and $x$ are morphisms, $x'''$ is also a morphism.
Furthermore, the commutativity of the top triangle is by construction of $x'''$.
Finally, let us prove the commutativity of the bottom right triangle.
For $j \leq k$ or ($j = k+1$ and $i \leq n_{k+1}$)
\[
	f\cdot x'''(i,j) = f\cdot x(i,j) = y(i,j) = y\cdot \JS s'''(i,j).
\]
For $j = k+1$ and $i > n_{k+1}$,
\begin{align*}
	f\cdot x'''(i,k+1) 
		& = f \cdot x''(n_{k+1}'+i,k+1)\\
		& = y''\cdot s''_{k+1}(n_{k+1}'+i-n_{k+1})\\
		& = y(n_{k+1}'+s''_{k+1}(n_{k+1}'+i-n_{k+1})-n_{k+1},k+1)\\
		& = y\cdot \JS s'''(i,k+1).
\end{align*}
Finally, for $j = k+2$
\[
	f\cdot x'''(i,k+2) = f\cdot x''(i,k+2) = y'\cdot \JS s''(i,k+2) = y\cdot \JS s'''(i,k+2).
\]
\qed
\end{proofappx}

\begin{proofappx}{lem:path-implies-ES}
Let us prove that $\pi_X$ is $(\E,\bbS)$-open. 
We then start with the following square
\[
  \begin{tikzcd}
            Jv
            \arrow{r}[description,inner sep=2pt]{r}
            \arrow{d}[description,inner sep=2pt]{\JE e}
            \descto{dr}{\commutes}
            & \widetilde{R}
            \arrow{d}[description,inner sep=2pt]{\pi_X}
            \\
            Jw
            \arrow{r}[description,inner sep=2pt]{x}
            & X
        \end{tikzcd}
  \]
and we can assume that $r$ is a strong morphism by 
Lemma~\ref{lem:restriction-to-strong}.
There are two cases, depending on whether $e = e_a$ or $e_\tau$.
Let us prove the case for $e_a$, the other case being similar.
So $v = n_1, a_1, \ldots, a_k, n_{k+1}$ and 
$w = n_1, a_1, \ldots, a_k, n_{k+1}', a, n_{k+2}$ with $n_{k+1} \leq n_{k+1}'$.
By Lemma~\ref{lem:run-R}, $\text{end}(r)$ is of the form
$(X \xleftarrow{~x'~} Jv \xrightarrow{~y'~} Y)$ and $\pi_X\cdot r = x'$.
In particular, $(X \xleftarrow{~x'~} Jv \xrightarrow{~y'~} Y)$ belongs to $R$.

In summary, we have the following situation (in plain)

	\[
	\begin{tikzcd}
		Jw
		\arrow{dd}[description,inner sep=2pt]{x}
		& \phantom{}
		\descto[pos=0.5]{d}{\commutes}
		& Ju
		\arrow[dashed]{ll}[description,inner sep=2pt]{\JS s}
		\arrow[dashed]{dd}[description,inner sep=2pt]{y}
		&\phantom{}\\
		\phantom{}
		\descto[pos=0.5]{r}{\commutes}
		& Jv
		\arrow{lu}[description,inner sep=2pt]{\JE e}
		\arrow[dashed]{ru}[description,inner sep=2pt]{\JE e'}
		\arrow{ld}[description,inner sep=2pt]{x'}
		\arrow{rd}[description,inner sep=2pt]{y'}
		\descto[pos=0.5]{r}{\commutes}
		\descto[pos=0.5]{d}{\ensuremath{\in R}\xspace}
		& \phantom{}
		& \in R\\
		X
		&\phantom{}
		& Y
		&\phantom{}
       	\end{tikzcd}
    	\]
and since $R$ is a path bisimulation, we obtain the rest of the diagram (dashed).
In particular, $(X \xleftarrow{~x\cdot \JS s~} Jv \xrightarrow{~y~} Y)$ belongs to $R$.
To conclude, it is enough to construct a morphism $r'\colon Ju\to\widetilde{R}$
such that
   \[
  \begin{tikzcd}
            Jv
            \arrow{rr}[description,inner sep=2pt]{r}
            \arrow[dashed]{dr}[description,inner sep=2pt]{\JE e'}
            \arrow{dd}[description,inner sep=2pt]{\JE e}
            \descto[pos=0.45]{drr}{\commutes}
            \descto[pos=0.75]{ddrr}{\commutes}
            & & \widetilde{R}
            \arrow{dd}[description,inner sep=2pt]{\pi_X}
            \\
            \phantom{}
            \descto[pos=0.5]{r}{\commutes}
            &Ju
            \arrow[dashed]{ur}[description,inner sep=2pt]{r'}
            \arrow[dashed]{dl}[description,inner sep=2pt]{\JS s}
            & \phantom{}
            \\
            Jw
            \arrow{rr}[description,inner sep=2pt]{x}
            & & X
        \end{tikzcd}
  \]
  
By construction, $u$ is of the form $n_1, a_1, \ldots, a_k, n_{k+1}'', a, n_{k+2}'$
with $n_{k+1} \leq n_{k+1}' \leq n_{k+1}''$ and $n_{k+2} \leq n_{k+2}'$.
For every $(i,j)$ state of $Ju$, the restriction of $Ju$ up to $(i,j)$ is obtained 
as $Ju_{i,j}$ where $u_{i,j} = n_1, a_1, \ldots, a_{j-1}, i$. Let us denote 
$e_{i,j}$ the injection of $Ju_{i,j}$ into $Ju$, which is a strict morphism as a composition
of morphisms of the form $\JE e_a$ and $\JE e_\tau$.
Then define $r'$ as follows:
\[
	r'(i,j) = (X \xleftarrow{~x\cdot \JS s\cdot e_{i,j}~} Ju_{i,j} \xrightarrow{~y\cdot e_{i,j}~} Y).
\]
First, $r'$ is with values in $\widetilde{R}$, because 
$(X \xleftarrow{~x\cdot \JS s\cdot e_{i,j}~} Ju_{i,j} \xrightarrow{~y\cdot e_{i,j}~} Y)$
belongs to $R$ by induction using the backward closure of $R$ being a strong path 
bisimulation.
It is also easy to see that $r'$ is a strict morphism.
For the equation $r'\cdot \JE e' = r$, this is a consequence of 
$x' = x\cdot \JE e = x\cdot \JS s\cdot \JE e'$.
Finally, $\pi_X\cdot r' = x\cdot \JS s$ is by construction and Lemma~\ref{lem:run-R}.
\qed
\end{proofappx}

\begin{proofappx}{lem:branching-stuttering-iff-strong-path}
Assume that $R$ is a stuttering branching bisimulation.
Let us prove that $\overline{R}$ is a strong path bisimulation.
The initial condition and the backward closure are obvious.
Let us then prove the forward closure. We then have the following situation (in plain)
	\[
	\begin{tikzcd}
		Jw
		\arrow{dd}[description,inner sep=2pt]{x'}
		& \phantom{}
		\descto[pos=0.5]{d}{\commutes}
		& Ju
		\arrow[dashed]{ll}[description,inner sep=2pt]{\JS s}
		\arrow[dashed]{dd}[description,inner sep=2pt]{y'}
		&\phantom{}\\
		\phantom{}
		\descto[pos=0.5]{r}{\commutes}
		& Jv
		\arrow{lu}[description,inner sep=2pt]{\JE e}
		\arrow[dashed]{ru}[description,inner sep=2pt]{\JE e'}
		\arrow{ld}[description,inner sep=2pt]{x}
		\arrow{rd}[description,inner sep=2pt]{y}
		\descto[pos=0.5]{r}{\commutes}
		\descto[pos=0.5]{d}{\ensuremath{\in \overline{R}}\xspace}
		& \phantom{}
		& \in \overline{R}\\
		X
		&\phantom{}
		& Y
		&\phantom{}
       	\end{tikzcd}
    	\]
and want to construct the dashed data.
Let us assume that $e = e_a$, the other case being simpler.
So $v$ is of the form $n_1, a_1, \ldots, n_{k+1}$ and $w$ 
of the form $n_1, a_1, \ldots, n_{k+1}',a,n_{k+2}$ with 
$n_{k+1} \leq n_{k+1}'$.
We also have:
\begin{itemize}
	\item $(x(n_{k+1},k+1),y(n_{k+1},k+1)) \in R$,
	\item $x(n_{k+1},k+1) = x'(n_{k+1},k+1)$ (since $x = x'\cdot \JE e$),
	\item for every $i \geq n_{k+1}$, ($x'(i+1,k+1) = x'(i,k+1)$ or 
		$x'(i,k+1) \xtransto{~\tau~} x'(i+1,k+1)$),
	\item $x'(n_{k+1}',k+1) \xtransto{~a~} x'(0,k+2)$,
	\item for every $i$, ($x'(i+1,k+2) = x'(i,k+2)$ or 
		$x'(i,k+2) \xtransto{~\tau~} x'(i+1,k+2)$),
\end{itemize}
the last three points holding since $x'$ is a morphism. Since $R$ is a branching bisimulation with the 
stuttering, we can build $(y_{i,j})_{n_{k+1} \leq j \leq n_{k+1}', 1 \leq i \leq m_j} $ states of $y$ such that:
\begin{itemize}
	\item $y_{1,n_{k+1}} = y(n_{k+1},k+1)$,
	\item for all $i$, $j$, $(x'(j,k+1),y_{i,j}) \in R$,
	\item $y_{i,j} = y_{i+1,j}$ or $y_{i,j} \xtransto{~\tau~} y_{i+1,j}$,
	\item $y_{m_j,j} = y_{1,j+1}$ or $y_{m_j,j} \xtransto{~\tau~} y_{1,j+1}$.
\end{itemize}
by induction on $j$.
\begin{itemize}
	\item \textbf{base case $j = n_{k+1}$:} $m_j = 1$, $y_{1,n_{k+1}} = y(n_{k+1},k+1)$.
	\item \textbf{inductive case:} $(y_{i,j})_{n_{k+1} \leq j \leq s, 1 \leq i \leq m_j}$ 
	has been constructed as wanted. 
	In particular, $(x'(s,k+1),y_{m_s,s}) \in R$.
	Two cases:
	\begin{itemize}
		\item $x'(s+1,k+1) = x'(s,k+1)$. 
		In that case, $m_{s+1} = 1$ and $y_{1,s+1} = y_{m_s,s}$.
		\item or $x'(s,k+1) \xtransto{~\tau~} x'(s+1,k+1)$. 
		Since $R$ is a stuttering branching bisimulation, two cases again:
			\begin{itemize}
				\item $(x'(s+1,k+1),y_{m_s,s})$. 
				In that case, $m_{s+1} = 1$ and $y_{1,s+1} = y_{m_s,s}$.
				\item or $y_{m_s,s} \xtransto{~\tau~} y_1 \xtransto{~\tau~} 
					\ldots \xtransto{~\tau~} y_d
					\xtransto{~\tau~} z_1 \xtransto{~\tau~} 
					\ldots \xtransto{~\tau~} z_{d'}$, with $(x'(s,k+1),y_{d''}) \in R$,
					$(x'(s+1,k+1),z_{d''}) \in R$, and $(x'(s+1,k+1),z_{d'}) \in R$ for all
					$d''$
					Then the following family 
					$(y_{i,j}')_{n_{k+1} \leq j \leq s+1, 1 \leq i \leq m'_j}$ works:
					\begin{itemize}
						\item for $j < s$, $m'_j = m_j$ and $y'_{i,j} = y_{i,j}$,
						\item $m'_s = m_s+d$, $y_{i,s}' = y_{i,s}$ 
						for $i \leq m_s$, and $y_{i,s}' = y_{i-m_s}$ otherwise,
						\item $m'_{s+1} = d'$ and $y_{i,s+1}' = z_i$.
					\end{itemize}
			\end{itemize}
	\end{itemize}
\end{itemize}
Repeating this process with the transition $x'(n_{k+1}',k+1) \xtransto{~a~} x'(0,k+2)$ and the 
transitions $x'(i,k+2) \xtransto{~\tau~} x'(i+1,k+2)$, we can build two families
$(y_{i,j})_{n_{k+1} \leq j \leq n_{k+1}', 1 \leq i \leq m_j}$ and
$(y_{i,j}')_{0 \leq j \leq n_{k+2}, 1 \leq i \leq m_j'}$ such that:
\begin{itemize}
	\item $y_{1,n_{k+1}} = y(n_{k+1},k+1)$,
	\item for all $i$, $j$, $(x'(j,k+1),y_{i,j}) \in R$,
	\item $y_{i,j} = y_{i+1,j}$ or $y_{i,j} \xtransto{~\tau~} y_{i+1,j}$,
	\item $y_{m_j,j} = y_{1,j+1}$ or $y_{m_j,j} \xtransto{~\tau~} y_{1,j+1}$,
	\item $y_{m_{n_{k+1}'},n_{k+1}'} \xtransto{~a~} y_{1,0}'$,
	\item for all $i$, $j$, $(x'(j,k+2),y_{i,j}') \in R$,
	\item $y_{i,j}' = y_{i+1,j}'$ or $y_{i,j}' \xtransto{~\tau~} y_{i+1,j}'$,
	\item $y_{m_j,j}' = y_{1,j+1}'$ or $y_{m_j,j}' \xtransto{~\tau~} y_{1,j+1}'$.
\end{itemize}
Now, define $u$ as $n_1, a_1, \ldots, 
	n_{k+1}+\sum\limits_{j = n_{k+1}}^{n_{k+1}'}m_j, a, \sum\limits_{j = 0}^{n_{k+2}}m_j'$.
By construction, there is a unique $e' = e_a\colon v\to u$ in $\E$.
Next $s = (s_1,\ldots,s_{k+2})$ is defined as follows:
\[
\begin{tabular}{m{4.7em}|m{2em}l}
	$s_j(i)$ = 
		& $i$
		& if $j \leq k$ or ($j = k+1$ and $i \leq n_{k+1}$)
		\\
		& $s$ 
		& if $j = k+1$ and \\
		&
		& $n_{k+1}+\sum\limits_{j = n_{k+1}}^{s-1}m_j < i \leq n_{k+1}
			+\sum\limits_{j = n_{k+1}}^{s}m_j$
		\\
		& $s$ 
		&if $j = k+2$ and 
		$\sum\limits_{j = 0}^{s-1}m_j' < i \leq \sum\limits_{j = 0}^{s}m_j'$.
\end{tabular}
\]
By construction, $s\colon u\to w$ is a morphism in $\bbS$ 
and $\JE e = \JS s\cdot \JE e'$. 
Finally, we can define $y'$ as
\[
\begin{tabular}{m{4.7em}|ll}
	$y'(i,j)$ = 
		& $y(i,j)$
		& if $j \leq k$ or ($j = k+1$ and $i \leq n_{k+1}$)
		\\
		& $y_{i-(n_{k+1}+\sum\limits_{j = n_{k+1}}^{s-1}m_j),s}$ 
		& if $j = k+1$ and \\
		&
		& $n_{k+1}+\sum\limits_{j = n_{k+1}}^{s-1}m_j < i \leq n_{k+1}
			+\sum\limits_{j = n_{k+1}}^{s}m_j$
		\\
		& $y_{i-(\sum\limits_{j = 0}^{s-1}m_j'),s}'$ 
		&if $j = k+2$ and 
		$\sum\limits_{j = 0}^{s-1}m_j' < i \leq \sum\limits_{j = 0}^{s}m_j'$.
\end{tabular}
\]
By construction, $y'\cdot \JE e' = y$.
It only remains to prove that 
$(X \xleftarrow{x'\cdot \JS s} Ju \xrightarrow{y'} Y) \in \overline{R}$, that is, 
for all $i$, $j$, $(x'\cdot \JS s(i,j),y'(i,j)) \in R$:
\begin{itemize}
	\item if $j \leq k$ or ($j = k+1$ and $i \leq n_{k+1}$), $x'\cdot \JS s(i,j) = x(i,j)$, 
	$y'(i,j) = y(i,j)$, and 
	$(x(i,j),y(i,j)) \in R$.
	\item if $j = k+1$ and $n_{k+1}+\sum\limits_{j = n_{k+1}}^{s-1}m_j < i \leq n_{k+1}
			+\sum\limits_{j = n_{k+1}}^{s}m_j$,
		$x'\cdot \JS s(i,j) = x'(s,k+1)$ and 
		$y'(i,j) = y_{i-(n_{k+1}+\sum\limits_{j = n_{k+1}}^{s-1}m_j),s}$, 
		and by construction, 
		$(x'(s,k+1),y_{i-(n_{k+1}+\sum\limits_{j = n_{k+1}}^{s-1}m_j),s}) \in R$.
	\item if $j = k+2$ and $\sum\limits_{j = 0}^{s-1}m_j' < i \leq \sum\limits_{j = 0}^{s}m_j'$,
		$x'\cdot \JS s(i,j) = x'(s,k+2)$ and 
		$y'(i,j) = y_{i-(\sum\limits_{j = 0}^{s-1}m_j'),s}$, 
		and by construction, 
		$(x'(s,k+2),y_{i-(\sum\limits_{j = 0}^{s-1}m_j'),s}) \in R$.
\end{itemize}

Conversely, assume that $R$ is a strong path bisimulation and let us prove that
$\widecheck{R}$ is stuttering branching bisimulation.
So we start with $(X \xleftarrow{x} Jv \xrightarrow{y} Y) \in R$ so that 
$(\text{end}(x), \text{end}(y))$ in $\widecheck{R}$.
Let us assume that $\text{end}(x) \xtransto{~a~} p$.
Let us do the case when $a = \tau$, the other case being similar.
Let us assume that $v$ is of the form $n_1, a_1, \ldots, n_{k+1}$.
Define $w = n_1, a_1, \ldots, n_{k+1}+1$ so that 
$e = e_\tau\colon v\to w \in \E$.
We can define $x'\colon Jw\to X$ as 
$x'(i,j) = p$ if $i = n_{k+1}+1$ and $j = k+1$; $x'(i,j) = x(i,j)$ otherwise.
This is a morphism since $x$ is a morphism and by the assumption on $p$.
Now, we have the following situation (in plain)
	\[
	\begin{tikzcd}
		Jw
		\arrow{dd}[description,inner sep=2pt]{x'}
		& \phantom{}
		\descto[pos=0.5]{d}{\commutes}
		& Ju
		\arrow[dashed]{ll}[description,inner sep=2pt]{\JS s}
		\arrow[dashed]{dd}[description,inner sep=2pt]{y'}
		&\phantom{}\\
		\phantom{}
		\descto[pos=0.5]{r}{\commutes}
		& Jv
		\arrow{lu}[description,inner sep=2pt]{\JE e}
		\arrow[dashed]{ru}[description,inner sep=2pt]{\JE e'}
		\arrow{ld}[description,inner sep=2pt]{x}
		\arrow{rd}[description,inner sep=2pt]{y}
		\descto[pos=0.5]{r}{\commutes}
		\descto[pos=0.5]{d}{\ensuremath{\in R}\xspace}
		& \phantom{}
		& \in R\\
		X
		&\phantom{}
		& Y
		&\phantom{}
       	\end{tikzcd}
    	\]
and since $R$ is path bisimulation, we obtain the dashed data.
In particular, $e' = e_\tau$ and $u$ is of the form
$n_1, a_1, \ldots, n_{k+1}'$ with $n_{k+1}+1 \leq n_{k+1}'$.
Two cases:
\begin{itemize}
	\item for every $i \geq n_{k+1}$, $y'(i,k+1) = \text{end}(y)$, then
		$(\text{end}(x'),\text{end}(y)) \in \widecheck{R}$,
	\item or there is a sequence $i_0 = n_{k+1} < i_1 < \ldots < i_m < n_{k+1}' = i_{m+1}$ such 
		that 
		\[
		\text{end}(y) \xtransto{~\tau~} y'(i_1,k+1) \xtransto{~\tau~} 
			\ldots \xtransto{~\tau~} y'(i_m,k+1) \xtransto{~\tau~}
			\text{end}(y')
		\]
		Now, there is necessarily $n_{k+1} \leq i' < n_{k+1}'$ such that
		\begin{itemize}
			\item for all $i \leq i'$, $\JS s(i,k+1) = n_{k+1}$,
			\item for all $i > i'$, $\JS s(i,k+1) = n_{k+1}+1$.
		\end{itemize}
		Let $l$ such that $i_l \leq i' < i_{l+1}$.
		If we denote $n_1, a_1, \ldots, i_h$ by $u_h$ for all $h$, 
		then $e_\tau\colon u_h\to u$ belongs to $\E$ and since  
		$(X \xleftarrow{x'\cdot \JS s} Jv \xrightarrow{y'} Y) \in R$
		and $R$ is backward closed, for all $h$,
		$(X \xleftarrow{x'_h = x'\cdot \JS s\cdot \JE e_\tau} Ju_h 
			\xrightarrow{y'_h = y'\cdot \JE e_\tau} Y) \in R$.
		Now:
		\begin{itemize}
			\item for $h \leq l$, 
				$(\text{end}(x),y'(i_h,k+1)) = (\text{end}(x_h),\text{end}(y_h) \in \widecheck{R}$,
			\item for $h > l$,
				$(\text{end}(x'),y'(i_h,k+1)) = (\text{end}(x_h),\text{end}(y_h) \in \widecheck{R}$.
		\end{itemize}
		which is what we had to prove.
\end{itemize}
\qed
\end{proofappx}

\begin{proofappx}{lem:weak-iff-path}
The constructions and the reasonings are the same as for 
\autoref{lem:branching-stuttering-iff-strong-path},
except that we do not have to care about intermediate states to be 
bisimilar (stuttering and backward closure).
\qed
\end{proofappx}

\end{document}
